\documentclass[3p,preprint,review]{elsarticle}
\usepackage{amsfonts}
\usepackage{amsmath}
\usepackage{amssymb}
\usepackage{amsthm}
\allowdisplaybreaks[1]
\usepackage{setspace}
\usepackage{algorithm}
\usepackage[normalem]{ulem}
\usepackage{amsmath,amsthm,amssymb,amsfonts, subfigure,  algorithm,algorithmic}
\usepackage[]{graphicx}

\usepackage{comment}
\usepackage{color}
\usepackage{longtable}

\newtheorem{theorem}{Theorem}
\newtheorem{definition}{Definition}
\newtheorem{lemma}{Lemma}
\newtheorem{observation}{Observation}

\begin{document}
%\linenumbers
\begin{frontmatter}
\title{An Efficient Heuristic for Betweenness-Ordering\tnoteref{t1}}
\tnotetext[t1]{This is an expanded and extended version of the results appeared in CompleNet 2015 \cite{Agarwal:2015}.}
\author[rrs]{Rishi Ranjan Singh \corref{cor1} \fnref{cor2}}
\ead{rishirs@iitrpr.ac.in}
%\author[rrs]{S.R.S. Iyengar}%\fnref{dg1} }
%\ead{sudarshan@iitrpr.ac.in}
\author[iitr]{Shubham Chaudhary}%\fnref{dg1} }
\ead{shubhuma@iitr.ac.in}
\author[iitr]{Manas Agarwal}%\fnref{dg1} }
\ead{manasuma@iitr.ac.in}
\cortext[cor1]{Corresponding author}
\fntext[cor2]{Present Address: Department of Information Technology, IIIT Allahabad, India}

\address[rrs]{Department of Computer Science and Engineering, Indian Institute of Technology Ropar, \\ Nangal Road,  Rupnagar, Punjab - 140001, India}
\address[iitr]{Department of Mathematics, Indian Institute of Technology Roorkee\\
Roorkee, Uttarakhand, India}

\begin{abstract}
%\begin{linenumbers}
%In conjunction with the big data problems there came the need to analyze big networks and in this connection, centrality measures became of great interest to the community of mathematicians, computer scientists, and physicists. While it is an essential question to ask how one can rank vertices based on their importance in a network, there hasn't been a commonly accepted definition, mainly due to the subjectivity of the term ``importance". 
Centrality measures, erstwhile popular amongst the sociologists and psychologists, have seen broad and increasing applications across several disciplines of late. Amongst a plethora of application specific definitions available in the literature to rank the vertices, closeness centrality, betweenness centrality and eigenvector centrality (page-rank) have been the most important and widely applied ones. Networks where information, signal or commodities are flowing on the edges, surrounds us. Betweenness centrality comes as a handy tool to analyze such systems, but betweenness computation is a daunting task in large size networks. In this paper, we propose an efficient heuristic to determine the betweenness-ordering of $k$ vertices (where $k$ is very less than the total number of vertices) without computing their exact betweenness indices. The algorithm is based on a non-uniform node sampling model which is developed based on the analysis of Erdos-Renyi graphs. We apply our approach to find the betweenness-ordering of vertices in several synthetic and real-world graphs. The proposed heuristic results very efficient ordering even when runs for a linear time in the terms of the number of edges. We compare our method with the available techniques in the literature and show that our method produces more efficient ordering than the currently known methods. 
% Extensive experimental evidence is presented to demonstrate the performance of the proposed heuristic.
%We further show that the accuracy of our algorithm gets better with the increase in size and density of the network.
%\end{linenumbers}
\end{abstract}
\end{frontmatter}
\pagebreak
\section{Introduction} 

The centrality of a vertex in a network is the quantification of the intuitive notion of the importance of a node in a network. In recent times, centrality measures have been extensively used in the analysis of large-scale real-world networks. Centrality indices, also referred as structural indices, are real-valued functions that remain invariant under isomorphic transformation of graphs \cite{Brandes:2005}. A number of application-specific centrality measures have been coined in the literature. For a detailed study of centrality indices and their applications, one can refer to the books by Newman~\cite{Newman}, Jackson~\cite{Jackson} and Brandes and Erlebach~\cite{Brandes:2005}. Real-world networks are usually gigantic, dynamic in nature and keep changing at a very high rate. In such networks, comparing centrality scores of two nodes is of great importance. Consider, for example, a production company is finalizing a new brand ambassador for their organization and has two options to choose. To evaluate which one is better, one might need to compare the importance (in this case popularity) of the two candidate actors in a given social network where the number of nodes is in the order of thousands. Consider two research papers in the citation network \cite{tang:2008} which is in the order of a few million nodes, how can one find which article is more central than the other one? For example, if one were to compute betweenness centrality in this case, even with the adoption of the best-known algorithm, it is a time-consuming task for large-sized networks. We ask this question \emph{ Is there a method to compute the centrality-ordering of two nodes and declare which one is more central than the other, without actually computing their exact centrality values}. More formally, given two nodes $u$ and $v$ in a graph $G$ with centrality values $C(u)$ and $C(v)$, if $C(u)>C(v)$ then $u$ is superior in rank to $v$. Can one get to know which node is of superior rank over the other without computing its centrality values? We call this problem the \emph{centrality-ordering} problem. Note that in this problem, nodes are arbitrarily given, and that is why this class of problems is different than the problem of finding top $k$ most central nodes \cite {Lim:2011,le:2014}.

In general, there are three possible ways to solve the centrality-ordering problem if we allow exact computation of centrality scores:
\begin{enumerate}
\item Compute the exact centrality scores of both the nodes and order them accordingly.
\item Efficiently approximate the centrality values of both the nodes and compare the scores to get the estimated ordering.
\item Directly compute the exact or approximate centrality-ordering exploiting some structural property of the given network without even calculating the individual centrality scores.
\end{enumerate}
The reason, we exclude the first type of solution for ordering problem is summarized below. 
This trivial method for exact centrality-ordering calculates the centrality score of both the nodes and then compares the values to answer which one is more important. There are two reasons why the current state of the art algorithms for exact calculation of the centrality measures are not time efficient. Firstly because of the large size and the dynamic nature of networks. In big dynamic networks, we have to recompute the centrality scores each time the network changes, which is evidently expensive. Secondly because of the global characteristics of some centrality measures. For example, closeness centrality and degree centrality computation of a single node takes very less time as compared to the calculation of the corresponding centrality measures for all the nodes in a network. But unlike degree and closeness centralities, computing betweenness centrality of a node is conjectured to be as expensive as computing it for all the nodes in any network \cite{Kintali:2008}. 

The second method efficiently estimates the individual centrality scores and tries to find the correct ordering quickly with a high probability. In the third and the last type of solution, we order the nodes based on some structural properties of the network, without computing or estimating the centrality scores. Such a kind of solution for a particular case of eccentricity-ordering is available in the appendix.

Networks with information, signal or commodity flowing on its edges are present everywhere in nature. Betweenness centrality is a popular tool to analyze such systems. Betweenness centrality was proposed by Freeman~\cite{Free77} and Anthonisse~\cite{Antho71} independently. Betweenness centrality of a node $v$ is defined as the relative fraction of shortest paths passing through $v$. It is calculated as $BC(v)=\sum\limits_{s\neq t\neq v \epsilon V}\frac {\sigma_{st}(v)}{\sigma_{st}}$, where $\sigma_{st}$ is the total number of shortest paths from vertex $s$ to vertex $t$ and $\sigma_{st}(v)$ is the total number of shortest paths from vertex $s$ to vertex $t$ passing through vertex $v$. Unlike degree centrality, betweenness centrality covers more global characteristics and unlike closeness centrality it works even on disconnected networks. Betweenness centrality has found many important applications in diverse fields. It has been used in biological networks \cite{Shiva05}, protein-protein interaction (PPI) networks \cite{Joy05}, analyzing communication system networks \cite{Tizghadam:2010}, identifying critical nodes in the electrical \& electronic systems (EES systems like Electronic Control Units used in vehicles) \cite{TCL}, analyzing supply chain networks \cite{Borgatti:2009}, identifying bottleneck in supply chain networks \cite{Mizgier:2012}, planning a better public transit system networks; for example metro networks \cite{Derrible:2012}, measuring load at a node in gas pipeline network \cite{Carvalho:2009}, waste-water disposal system networks etc. 

Since computing the betweenness centrality of one node is equivalent to computing the betweenness centrality of all nodes according to the currently known deterministic algorithms, we are motivated to address the problem of betweenness-ordering of two vertices. First, we give a heuristic to find the betweenness-ordering of two vertices (here onwards called the betweenness-ordering-problem). Then, we extend the algorithm for comparing $k$ ( $2<k\ll n$ ) vertices, where $n$ is the total number of nodes. In this paper, we propose a second type of solution for betweenness-ordering problem based on a non-uniform sampling based efficient estimation. We present expanded version of our estimation heuristic\cite{Agarwal:2015} and then discuss the betweenness-ordering results based on it. The algorithm uses a novel non-uniform sampling technique which approximates the optimal sampling model noted by Chehreghani~\cite{Chehreghani:2014} better than the non-uniform sampling model proposed by him. The betweenness score of a given node is estimated using the proposed sampling model incorporated within the approximation algorithm, given in section 3.

The contributions of this paper are as follows.
\begin{enumerate}
\item To the best of our knowledge, it is the first study that focuses on the ordering of nodes based on centrality scores. Several studies exist that target to find the top central nodes in a graph, but in our knowledge, no study explores the problem of finding the ranking two or more (randomly picked) nodes based on a particular centrality measure without computing their exact centrality scores.
\item We devise an efficient heuristic for betweenness-ordering of two nodes and further generalize it for ordering an arbitrary $k$ nodes ($k\ll$ total number of nodes). First, we discuss a very efficient non-uniform sampling technique to choose the source nodes (also called pivot nodes) for single source shortest path computation. Then we use the model to estimate the betweenness score of a given node efficiently without computing betweenness of all the nodes in the graph. 
\item The developed non-uniform node sampling model provides a better approximation of the optimal sampling model given by Chehreghani~\cite{Chehreghani:2014} than the non-uniform sampling model implemented by him.
\item We conduct extensive simulations with the real-world and synthetic networks. The results show that our heuristic and the proposed sampling model outperforms most of the sampling based and deterministic approximation algorithms for ordering the nodes based on betweenness score or estimating a node's betweenness.
\item Although the basis of the proposed heuristic is an analysis of random $G(n,p)$ graphs which are different than the real-world graphs, the reliability of the heuristic is made evident from the efficient performance on a wide range of several real-world networks picked from SNAP dataset \cite{snap}.         
\end{enumerate}  

We organize the rest of the paper as follows. In next section, we briefly discuss the algorithms employed for betweenness centrality computation. In section 3, we define basic terms used in the paper and explain the previous concepts, based on which we develop our sampling model. In section 4, we develop our model based on the analysis of random networks and some observations. Betweenness-ordering heuristic is discussed in section 5. All the details about simulations, data sets used in simulations, performance tools used for evaluation and comprehensive results in the form of plots and tables are compiled in section 6. We discuss the possible future directions of work and conclude the paper in section 7.   

\section{Related work} 

The ordering of nodes in a network based on betweenness centrality can find several applications in various real-world scenarios. To the best of our knowledge, it is the first work which considers and motivates the study on the betweenness-ordering of two nodes. Most of the studies done so far consider either computation of betweenness scores or ranking all the nodes based on their betweenness score. We summarize few of such algorithms in brief. Most of the exact algorithms for betweenness calculation are based on either single source shortest path (SSSP) computation algorithms from all sources or all pair shortest path computation algorithms. The most trivial algorithm is a modified version of the Floyd-Warshall's APSP algorithm \cite{Floyd,Warshall} to compute the betweenness scores for all nodes \cite{Free77}, but this takes O($n^3$) time where $n$ is the number of nodes. In the year 2001, Brandes \cite{Brandes:2001} introduced an algorithm based on Dijkstra' algorithm \cite{Dijkstra} which computes the exact betweenness score of all nodes in unweighted graphs in O($mn$) time, where $m$ is the number of edges.   

Even the state of the art (Brandes') algorithm is expensive in terms of time for large-sized real-world networks. This drawback motivated the researchers to develop faster exact or approximation algorithms. Several exact algorithms for large graphs (Sariy{\"u}ce et al.~\cite{Sariyüce:2013}) and dynamic graphs (Lee et al.~\cite{QUBE:2016}, Green et al.~\cite{Green12}, Kas et al.~\cite{Kas13}, Goel et al.~\cite{Goel:2013}, Nasre et al.~\cite{Nasre}) have been developed. These algorithms improved the computation time experimentally on a special type of graphs, but in the worst case, they all were as expensive as Brandes'~cite {Brandes:2001}. Several approximation algorithms were also proposed. These algorithms ran much faster and computed centrality scores close to the exact centrality scores. Two class of approximation algorithms exist. The First type consists of algorithms that focus on estimating the betweenness score of all the nodes together. The second type comprises of algorithms that approximate the betweenness score of a given node. Another class of algorithms is proposed \cite{Bergamini:2014,Bergamini:2015} that attempt to approximate betweenness centrality in dynamic graphs. A divide and conquer algorithm is given recently by Erdos et al.~\cite{erdos:2014} which computes betweenness centrality of nodes considering only the shortest paths between a set of target nodes. Our goal is to develop an efficient estimation algorithm to approximate betweenness score of a node. Therefore, we summarize most of the approximation ideas developed so far for betweenness computation.

Eppstein and Wang~\cite{Eppstein} first proposed the concept of sampling to compute approximately the centrality indices for which SSSP computation is required from all the nodes. They suggested to compute SSSP from only a few nodes (called pivot) and discussed how to approximate the closeness centrality. Brandes and Pich~\cite{Brandes:2007} extended the idea of sampling given by Eppstein and Wang for approximating betweenness centrality. They gave different pivot selection strategies. SSSP from each pivot node were computed to estimate the contribution of each pivot node in the betweenness score of all nodes. By extrapolating the average contribution from pivot nodes, the betweenness centrality was estimated. 

Bader et al.~\cite{Bader:2007} proposed an adaptive sampling based approximation algorithm to calculate the betweenness score of a given node. In his study, uniform probabilities were considered to sample the nodes. The number of sampled nodes were dependent on the importance of the considered node, i.e., for highly central nodes, the algorithm requires to sample less number of nodes as compared to the nodes sampled for less central nodes. They also provided a theoretical bound for their approximation algorithm. Geisberger et al.~\cite{Geisberger:2008} generalized the approach given by Brandes and Pich~\cite{Brandes:2007} and observed that the betweenness centrality scores of unimportant (less betweenness central) nodes which are near to pivot nodes, get overestimated. They provided an unbiased betweenness estimator framework which overcomes the observed problem. Gkorou et al.~\cite{Gkorou} developed two approximation approaches to estimate betweenness. Their first approach was for the dynamic networks and was based on the observation that very highly central nodes remain almost invariant over dynamic operations. The second algorithm was for large networks and considered only $k$-length shortest paths for the computation of approximate betweenness score. 
%Their algorithms did not perform well on random graphs.
Riondato and Kornaropoulos~\cite{Riondato} recently developed two randomized algorithms to approximate betweenness score based on a sampling of shortest paths and analyzed theoretically. The first algorithm approximates the betweenness score for all the nodes and the second algorithm approximates the betweenness score for top-$k$ nodes.
% They provided theoretical bounds using VC-Dimension theory for both algorithms. 

Recently, Chehreghani~\cite{Chehreghani:2014} proposed a new idea of approximating the betweenness score of a given node. He used non-uniform sampling and then unlike \cite{Brandes:2007, Geisberger:2008}, he scaled the contributions from sampled nodes with respect to the probabilities. Finally, he averaged all the scaled values to achieve the approximate score. He used a very trivial model for generating the non-uniform probabilities without giving any theoretical derivation.
 
\section{Preliminary}
In this section, we introduce some basic terms related to the betweenness centrality which has been used throughout the paper. We also discuss a recent concept that gives motivation for our sampling technique.
\subsection{Terminology}
We use following terms interchangeably; node or vertex and graph or network. For simplicity, we consider only unweighted undirected graphs. All the concepts discussed in this paper can be easily extended for weighted or directed graphs. Given a graph $G=(V,E)$, $V$ is the set of nodes with $|V|=n$ and $E$ is the set of edges with $|E|=m$. A (simple) $path$ is a sequence of edges connecting a sequence of vertices without any repetition of vertices. The $length$ of a path is the number of edges in the path. \textit{Shortest paths} between two vertices are the smallest length paths between them. \textit{Distance} between two nodes $i$ and $j$, $d(i,j)$, is the length of shortest path between $i$ and $j$.

Let $\sigma_{st}$ be the number of shortest paths between $s$ and $t$, for $s,t\in V$. Let $\sigma_{st}(v)$ be the number of shortest paths between $s$ and $t$ passing through $v$, for $v\in V$. Betweenness centrality score of a node $v\in V$ is calculated as \[BC(v)=\sum\limits_{s\neq t\neq v \epsilon V}\frac {\sigma_{st}(v)}{\sigma_{st}}.\] $Pair \;dependency$ of a pair of vertices $(s,t)$ on a vertex $v$ is defined as: $\delta _{st}(v) = \frac{\sigma_{st}(v)}{\sigma_{st}}$. Betweenness centrality of a vertex $v$ can be defined in terms of pair
dependency as \[BC(v)=\sum\limits_{s\neq v\neq t \in V}\delta _{st}(v).\] Let $BFT_r$ denotes the breadth-first traversal (BFT) of the graph rooted on vertex $r$. In $BFT_r$, we assume that $r$ is at level $0$ and the next levels are labelled by natural numbers in an increasing order.
 \textit{Dependency} of a vertex $s$ on a vertex $v$ is defined as: $\delta_{s\bullet}(v)=\sum\limits_{t \in V\setminus \{s,v\}}\delta _{st}(v)$.
Let us define a set $P^s(w)$ = $\{v:\; v\in V,\; w \; is \; a \; successor\; of \; v\; in \: BFT_s\}$. Brandes~\cite{Brandes:2001} proved that:
\begin{equation}
\delta_{s\bullet}(v)=\sum\limits_{w: v \in P^s(w)} \frac{\sigma_{sv}}
{\sigma_{sw}}(1 + \delta_{s\bullet}(w)).
\label{eq3}
\end{equation}
          
\subsection{A Betweenness Estimation Technique Based on Non-uniform Sampling \label{DBM}}
In this section, we briefly describe the recent work of Chehreghani~\cite{Chehreghani:2014}. An improvement on this work is given by us in this paper in section 4. Chehreghani gave an approximation algorithm to compute betweenness score of a given node $v$. The algorithm is summarized as Algorithm~\ref{BO_algo1}. For a given node $v$, the algorithm takes the sampling probabilities as input and outputs the approximate betweenness score of node $v$. Step 2 initializes the betweenness score to 0. The algorithm estimates the betweenness score of a node $v$, a $T$ number of times and takes the average of all the $T$ estimations. In each iteration of the algorithm, it samples a pivot node and computes the dependency of the pivot node on node $v$ using a single iteration of Brandes' algorithm \cite{Brandes:2001}. Then it estimates the betweenness score of node $v$ by, dividing (scaling) the computed dependency by the sampling probability of that pivot node. He has motivated his paper with the idea of optimal sampling that is stated in the following theorem.
\begin{algorithm}[t]$Estimate (G,P,v,T)$
\begin{algorithmic}[1]
\STATE \textbf{Input.} Graph $G$, probabilities $P=\{p_1,p_2, \cdots, p_n\}$, node $v$, number of samples $T$.
\STATE $BC(v)=0$.
\FOR {i=1 to T}
\STATE Select a node $i$ with probability $p_i$.
\STATE Compute $\delta_{i\bullet}(v)$ in the $BFT_{i}$ using Equation~\ref{eq3}.
\STATE $BC(v) \leftarrow BC(v)+\frac{\delta_{i\bullet}(v)}{p_i}$.
\ENDFOR
\STATE $BC[v]\leftarrow BC(v)/T.$
\STATE \textbf{Return.} $BC(v)$.
\end{algorithmic}
\caption{: Estimation algorithm to compute betweenness score of a given node $v$ \cite{Chehreghani:2014}. }
\label{BO_algo1}
\end{algorithm}
\begin{theorem}\cite{Chehreghani:2014} Let the sampling probability assigned to each node $i$ be    

\[
p_i = \dfrac{\delta_{i\bullet}(v)}{\sum_{j=1}^{n}\delta_{j\bullet}(v)}
\]
then, betweenness score of node $v$ can be exactly calculated in $O(m)$ time using single iteration of Algorithm~\ref{BO_algo1}.
\label{thrm01}
\end{theorem}

We refer the probability defined in Theorem~\ref{thrm01} as \textit{optimal probability} and call a model \textit{optimal model (OPT)} if it can generate optimal probabilities. Calculating optimal probabilities are as expensive as computing exact betweenness using Brandes' algorithm \cite{Brandes:2001}. Thus, a model was desired that can efficiently estimate sampling probabilities close to the optimal. Chehreghani noted that any such model should satisfy at least the following relation for most of the vertex pairs $(i, j)$:

\begin{equation}
p_{i}<p_{j} \Longleftrightarrow \delta_{i\bullet}(v)< \delta_{j\bullet}(v)
\end{equation}
Chehreghani has given a simple distance based model (DBM)\cite{Chehreghani:2014} to generate the sampling probabilities. He proposed to take the probabilities as the normalized value of the inverse of distance from node $v$ to node $i$,
$p_i \propto \frac{1}{d(v,i)}.$
He has shown experimentally that his non-uniform sampling technique reduces the error in the computation of betweenness score as compared to uniform sampling technique due to Brandes and Pich~\cite{Brandes:2007} and Bader et al.~\cite{Bader:2007}. But, Chehreghani was unable to provide a theoretical derivation for DBM. In DBM, many of the nodes $j$ with $\delta_{j\bullet}(v)= 0$ get same probabilities as nodes $i$ with $\delta_{i\bullet}(v)\neq 0$ because of being at the same level in $BFT_{v}$. We next propose a new probability estimation model for nodes that efficiently approximates the optimal probabilities and outperforms DBM. 
 
\section{ A New Non-uniform Sampling Model}

We consider the problem of time-efficiently ordering of two nodes based on their betweenness centrality. To order the two nodes based on their betweenness scores, we first efficiently estimate their individual betweenness scores and then compare. Our main problem reduces to a sub-problem which requires computing a very efficient approximation of the betweenness score of a given node. In this section, we discuss a model which generates non-uniform probabilities for sampling the nodes. This model can be incorporated with Algorithm~\ref{BO_algo1} to solve the above sub-problem. Our model is based on the inverse of degree and an exponential function in the power of distance. Thus we refer it as EDDBM (Exponential in Distance and inverse of Degree Based Model). The developed model reduces the average error by generating probabilities very close to the optimal probabilities. We try assigning larger probability values to the vertices contributing more to the betweenness of a given node $v$ and smaller to those which contribute less. We analyze random $G(n,p)$ graphs to establish the relation between the node sampling probabilities and the distance between the considered node and node to be sampled. Then, based on few observations, we propose a relation between the sampling probability and the degree of the nodes to be sampled. Finally, the steps to generate the probabilities by our model is described. 

The main reason to pick random $G(n,p)$ graphs for analysis is its non-complex characteristic. Analysis of other categories of synthetic graphs, for example, scale-free graphs and small world graphs is hard due to their complex nature. Though, a model developed based on the random graphs may not be widely applicable on real-world graphs, we prove the worthiness of our model by testing it on a large set of real-world graphs.  
\subsection{Analysis of Random Graphs}
\par\noindent
Let $G$ be a random graph that is generated based on the $G(n,p)$ model given by Erdos and Renyi~\cite{ER}. We are given a vertex $v$ to compute its betweenness score. We first analyze how the dependency of a node $i$ on the node $v$, $\delta_{i\bullet}(v)$ varies when $v$ lies on different levels in $BFT_i$. This will help us to establish a relation between $\delta_{i\bullet}(v)$ and the distance between $i$ and $v$. For this, first, we need to compute the expected number of nodes at any level $m$ of a BFS traversal. Wang~\cite{Wang} gave a complex approach to estimate the number of nodes at any level in BFS traversal on various types of graphs based on generating functions and degree distribution. We discuss a simple approach to estimate the number of nodes at a level in BFS traversal in random $G(n,p)$ graphs. Let $\lambda$ be the average degree of the given graph and let $p$ be the probability of an edge's existence. The first lemma approximately estimates the number of nodes at a given level in a BFS traversal by a recurrence relation.\\

\begin{lemma} Let $\alpha_{j}$ be the number of nodes at level $j$ in the $BFS_i$. Then the number of nodes at level $m+1$, $\alpha_{m+1}$ can be given as: 
\begin{equation}
\alpha_{m+1}\approx np(1-\dfrac{\sum_{j=0}^{m} \alpha_{j}}{n}) \alpha_{m}.
\label{eq1}
\end{equation}
\label{l1}
\end{lemma}
\begin{proof} Van Der Hofstad~\cite{Hofstad} explained the BFS traversal as Exploration Technique (ET) in random graphs. In this technique, all the vertices are initially inactive except $i$ (the root node on which ET has to be applied). The vertex $w$ is chosen which was discovered first among the current active vertices, and its neighborhood is explored for inactive vertices. All the inactive vertices found are marked active. Node $w$ is made inactive and is labeled as processed. In the paper we refer exploring the neighbourhood of $t^{th}$ vertex as $t^{th}$ exploration.

The following variables are same as in the exploration technique due to Van Der Hofstad~\cite{Hofstad}. Let $S_{t}$ be the number of active vertices after $t^{th}$ exploration, and $X_{t}$ be a random variable that denotes the number of vertices discovered (converted from inactive to active) in the $t^{th}$ exploration. Then, following relation holds for any iteration (exploration) $t$:
\begin{equation}
S_{t} = S_{t-1}+X_{t}-1
\end{equation}
After $t-1$ explorations, we are left with $n-(t-1)-S_{t-1}$ vertices ($t-1:processed$ $vertices$, $S_{t-1}:active$ $vertices$). If $p$ is the probability of existence of an edge between any two nodes in the graph then conditionally $S_{t-1}$ on we have:
\begin{equation}
X_{t}\sim Bin(n-(t-1)-S_{t-1},p)
\end{equation}
according to Equation~4.1.4 in \cite{Hofstad}.
Next, we state a very well known binomial relation in mathematics that we use to compute the expected number of nodes at any level of BFS traversal. If $X$ $\sim Bin(n,p)$ then, $Pr(X=k)=\binom{n}{k} p^{k} (1-p)^{n-k}$. The expected value of $X$, $E[X]$ is:
\[
E[X]=\sum_{k=1}^{n} k \binom{n}{k} p^{k} (1-p)^{n-k} = np.
\]

Using above relation, we can write the expected value of $X_t$ conditioned on $S_{t-1}$ as:
\begin{equation}
 E[X_{t}]=(n-(t-1)-S_{t-1})p
 \label{eq11}
\end{equation}

Equation~\ref{eq11} can be used in the following way to calculate the expected number of nodes at a BFS level. Initially, there is a single (source) node as an active node, i.e., $S_{0}=1$. The expected number of nodes discovered in the first exploration will be:
\[
E[X_{1}]=(n-1)p
\]

After the first exploration, the total number of active nodes is $S_1=X_1$. Therefore, the expected number of active nodes after the first exploration will be $(n-1)p$. The expected number of nodes discovered in the second exploration will be $(n-1-(n-1)p)p$ or 

\[
 E[X_{2}]=(n-1)(1-p)p.
\]

Similarly, we can calculate following values: $E[X_{3}]=(n-1)(1-p)^{2}p$, $ E[X_{4}]=(n-1)(1-p)^{3}p$. In the above manner, we can also calculate the number of active nodes before each exploration and the expected number of nodes discovered in each exploration. 

Let $\alpha_{j}$ be the expected number of nodes at level j. We have $\alpha_{0}=1$ and $\alpha_{1}=(n-1)p$. Then by using Equation~\ref{eq11} we can calculate $\alpha_2$ as:

\begin{equation}
\alpha_{2}=\sum_{k=1}^{(n-1)p} (n-1)(1-p)^{k}p=(n-1)(1-p)[1-(1-p)^{(n-1)p}]
\end{equation}
Now, we can derive the formula for the general case. Let us assume that $m-1$ levels have been explored, i.e the nodes of level $m$ have been discovered. Now we have to explore the nodes of level $m$. At this step, the expected number of undiscovered node  is $(n-\sum_{j=0}^{m} \alpha_{j})$. As $p$ is a uniform probability for the existence of an edge between two nodes, exploring the first vertex of level $m$ discovers $(n-\sum_{j=0}^{m} \alpha_{j})p$ expected number of nodes for level $m+1$. Exploration of the next vertex discovers $[(n-\sum_{j=0}^{m} \alpha_{j})-(n-\sum_{j=0}^{m} \alpha_{j})p]p=(n-\sum_{j=0}^{m} \alpha_{j})(1-p)p$ expected number of nodes and so on. So the expected number of nodes at level $m+1$ will be:
\[
\alpha_{m+1}=\sum_{k=0}^{\alpha_{m}-1} (n-\sum_{j=0}^{m} \alpha_{j})(1-p)^{k} p
\]
or
\begin{equation}
\alpha_{m+1}=(n-\sum_{j=0}^{m} \alpha_{j})[1-(1-p)^{\alpha_{m}}]
\label{eq12}
\end{equation}
 
Most of the real-world graphs around us are sparse. Therefore, to take analyze random graphs with similar characteristics, let us assume that $ p \ll 1$. Applying Binomial expansion and neglecting higher order terms of p, we can rewrite Equation~\ref{eq12} as:
\[
\alpha_{m+1}\approx (n-\sum_{j=0}^{m} \alpha_{j}) \alpha_{m} p=n(1-\dfrac{\sum_{j=0}^{m} \alpha_{j}}{n}) \alpha_{m} p.
\]
\end{proof} 

Equation~\ref{eq1} is a recurrence relation to estimate the number of nodes at some level $m+1$. Using Lemma \ref{l1}, we can estimate the ratio between the expected number of nodes at two consecutive levels. The ratio is derived as follows.\\

In random graphs, the average degree $\lambda$ is equal to $(n-1)p$, where $n$ is the number of nodes and $p$ is the existential probability of an edge. $\lambda$ can be approximated as $np$ for large $n$. If we denote $(1-\dfrac{\sum_{j=0}^{m} \alpha_{j}}{n})$ as $c_{m+1}$ (the fraction of nodes below level $m$), then we can rewrite Equation~\ref{eq1} as $\alpha_{m+1}\approx c_{m+1} \lambda \alpha_{m}$
or
\begin{equation}
\dfrac{\alpha_{m+1}}{\alpha_{m}}\approx c_{m+1} \lambda 
\label{eq2}
\end{equation}
where $c_{m+1}\in[0,1)$.

Based on Equation~\ref{eq2}, we derive the formula to calculate the expected dependency of a node $i$ on node $v$, $E[\delta_{i\bullet}(v)]$ in next lemma. Then, we establish the ratio between the expected dependency of root node $i$ on two nodes at consecutive levels in Theorem~\ref{t3}.\\

\begin{lemma} Let $l$ be the last level in $BST_i$. Let $v$ be a node at $(l-2)$th level in the $BFS_i$. Then the expected dependency of node $i$ on node $v$ can be given as 
\begin{equation}
E[\delta_{i\bullet}(v)]\approx c_{l-1}\lambda(1+c_{l}\lambda).
\label{eq4}
\end{equation}
\label{l2}
\end{lemma}

\begin{proof}
Let the BFS traversal rooted at $i$ consist of $l+1$ levels. If $v$ is at the last level ($l$), then $\delta_{i\bullet}(v)=0$. Now, if $v$ lies at level $l-1$, then we can compute the expected dependency ($E[\delta_{i\bullet}(v)]$) as follows. Let $A_{l-1}$ be the expected number of paths of length $l-1$ from $i$ to any vertex at level $l-1$. Bauckhage et al.~\cite{Bauckhage} gave following expression for $A_{l-1}$ :
\begin{equation}
A_{l-1}=n^{l-2}\pi^{l-1}
\end{equation}
where $\pi=p \dfrac{n-1}{n}$. It is easy to observe that any node $w$ at level $l$ has $\alpha_{l-1}\cdot p$ expected number of parents (nodes at level $l-1$ which are connected to $w$ by a direct edge), so the expected number of shortest paths from $i$ to $w$ will be $A_{l-1}\cdot \alpha_{l-1}\cdot p$. Node $v$ lies only on $A_{l-1}$ expected number of shortest paths out of those shortest paths. From Equation~\ref{eq3}, it is easy to observe that the expected partial dependency from node $i$ to node $w$ on node $v$ is $E[\delta_{iw}(v)]=\dfrac{1}{\alpha_{l-1} p}$. Node $v$ has $\alpha_{l} \cdot p$ children similar to $w$. Therefore, the expected dependency of node $i$ on node $v$ is $E[\delta_{i\bullet}(v)]=\dfrac{\alpha_{l}}{\alpha_{l-1}}$ or using  Equation~\ref{eq2} we can rewrite:
\[E[\delta_{i\bullet}(v)]\approx c_{l}\lambda.\]
Similarly, if $v$ lies at level $l-2$, then the expected dependency of node $i$ on node $v$ can be given as: $$E[\delta_{i\bullet}(v)]=(\dfrac{\alpha_{l-1}}{\alpha_{l-2}})(1+c_{l}\lambda)\approx c_{l-1}\lambda(1+c_{l}\lambda).$$ 
\end{proof}
Now, we can give the theorem stating the ratio between dependencies of root node on two nodes positioned at two consecutive levels. 

\begin{theorem} Let $l$ be the last level in $BST_i$. Let $\delta_{i\bullet}(v_{l-k})$ be the dependency of node $i$ at a node $v_{l-k}$ at level $l-k$ and let $\delta_{i\bullet}(v_{l-k+1})$ be the dependency of node $i$ at a node $v_{l-k+1}$ at level $l-k+1$. Then we have
\begin{equation}
\dfrac{E[\delta_{i\bullet}(v_{l-k})]}{E[\delta_{i\bullet}(v_{l-k+1})]}=c_{l-k+1}(\dfrac{1}{\phi} + \lambda)
\label{eq5}
\end{equation}  
where $\phi=(c_{l-k+2})(1+c_{l-k+3}\lambda(1+c_{l-k+4}\lambda(1+c_{l-k+5}\lambda(1+\cdots(1+c_{l}\lambda))\cdots).$
\label{t3}
\end{theorem}

\begin{proof}
The ratio of expected dependencies of node $i$ on $v$, when $v$ lies at level $l-2$, ($E[\delta_{i\bullet}(v_{l-2})]$) to when $v$ lies in level $l-1$, ($E[\delta_{i\bullet}(v_{l-1})]$) is
\begin{equation}
\dfrac{E[\delta_{i\bullet}(v_{l-2})]}{E[\delta_{i\bullet}(v_{l-1})]}=\dfrac{c_{l-1}}{c_{l}} (1+c_{l}\lambda).
\end{equation}
In general, the ratio of the expected dependencies for two successive levels $l-k$ and $l-k+1$ can be given as Equation~\ref{eq5}. 
\end{proof}

It is easy to observe that $c_{m}$ decreases continuously as $m$ increases. As $v$ becomes one level closer to $i$, the expected dependency of $i$ on $v$, $E[\delta_{i\bullet}(v)]$ increases proportional to the average degree $\lambda$. Therefore, on the basis of Theorem~\ref{t3}, we can assign a probability $p_{i}$ to the node $i$ defined as following.\\
\begin{definition} Suppose, we have to compute the betweenness score of node $v$. Then the sampling probability assigned to node $i$ is :
\begin{equation}
p_{i}\propto (\lambda)^{-d(i,v)}
\label{eq15}
\end{equation}
where $d(i,v)$ is the distance between $v$ and $i$.
\end{definition}

\subsection{Further Tweak}
In this section, we discuss some of the observations and propose some possible solutions to tackle an observed problem. In $BFT_v$, nodes at the same level are called siblings. We define successors of a node $j$ in $BFT_v$, $Succ_v(j)$, as the set of nodes to which at least one shortest path from $v$ passes through $j$. Similarly, we define predecessors of a node $j$ in $BFT_v$, $Pred_v(j)$, as the set of predecessors. Let $Reach_j^v$ be the set of nodes that are at most as far as $v$ from $j$.\\

\begin{observation} In the BFS tree rooted at the given node $v$, siblings get equal probabilities by Equation~\ref{eq15}, but might not contribute equally in the betweenness of $v$.
\end{observation}
\begin{comment}
\begin{figure}
\centering
\includegraphics[height=4cm]{im1.jpg}
\caption{An example $BST_{v}$}
\label{fig1}
\end{figure}
\end{comment}

For example, consider a connected undirected graph with $n$ nodes and following attributes. Let $v$ be a considered node in the graph for estimating betweenness score. Let $v$ be the neighbor of two nodes $i$ and $j$, where $j$ is a terminal node (node with degree 1). Then, $v$ is connected to the rest of $n-3$ nodes via $i$. When the BFS traversal rooted at node $v$ is drawn, node $v$ takes the place of the root node and $i$ and $j$ falls at level $1$. At $i$, a subtree with $n-2$ nodes hangs while $j$ does not have any child. Due to lying at the same level in the BFS traversal, according to Equation~\ref{eq15}, equal probabilities will be assigned to both $i$ and $j$. But, $\delta_{i\bullet}(v)=1$ and  $\delta_{j\bullet}(v)= n-2>>1$. Thus we need to tweak the formula to resolve this problem.\\

\begin{observation} In $BST_i$, no node from $Succ_v(i) \cup Pred_v(i) \cup Reach_i^v$ will contribute in $\delta_{i\bullet}(v)$.
\end{observation}
\begin{comment}
\begin{figure}[h]
\centering
\subfigure[]{%
\includegraphics[width=3.5cm]{im2.jpg}
\label{fig2}}
\quad
\subfigure[]{%
\includegraphics[width=3cm]{im3.jpg}
\label{fig3}}
\caption{ }
%\label{fig4}
\end{figure}
Figure~\ref{fig2} is $BST_{v}$ and we have $Succ_v(i)=dt_i+st_{ij}$. Figure~\ref{fig3} is $BST_{i}$.
It is easy to observe that none of the shortest paths from $i$ to the nodes in $Succ_v(i)$ will ever pass through node $v$. Thus these nodes will not contribute in $\delta_{i\bullet}(v)$.
\end{comment}

Observation 2 infers that in $BFT_v$, a node $i$ with larger number of successors will contribute ($\delta_{i\bullet}(v)$) lesser, i.e., the relation can be assumed as $\delta_{i\bullet}(v)\propto \frac{1}{|Succ_v(i) \cup Pred_v(i) \cup Reach_i^v|}$. Then based on the assumption, the probability assigned to node $i$ should also satisfy the following relation:
\[p_{i}\propto \frac{1}{|Succ_v(i) \cup Pred_v(i) \cup Reach_i^v|}.\]

Maintaining the sets $Succ_v(i)$, $Pred_v(i)$, and $Reach_i^v$ for each node in $BFT_v$ (graph) can not be achieved in linear time. At place of $Succ_v(i) \cup Pred_v(i) \cup Reach_i^v$, we use degree of node $i$. One reason for using degree is that it can be linearly computed and most of the time, it is the best predictor for the number of successors, predecessors in random graphs. Another reason is the high correlation between the Betweenness centrality and the Degree centrality \cite{Valente, Lee:2006}. 

\begin{comment}
It means a node with high degree is most probable to be a high betweenness central node. It is also easy to observe that if a node is high betweenness central then in the BFS tree rooted at that node, it may have many other optional nodes (siblings) of $v$ to reach other nodes below $v$.
\end{comment}
Thus, to overcome the problem stated in Observation 1, we also include the following relation:  

\begin{equation}
p_{i}\propto \frac{1}{deg(i)}
\label{eq16}
\end{equation}
where $deg(i)$ is the degree on node $i$ in the given graph. The final relation can be written as: 
\[p_{i}\propto \frac{ (\lambda)^{-d(v,i)}}{deg(i)}.\] 

We use the distance as the inverse power of the exponential function over the average degree, and the inverse of degree for modeling the probability generation. Thus we name this model \textit{EDDBM (exponential in the inverse of distance and inverse of degree based model)}. Next, we discuss the steps to generate the probabilities according to EDDBM.

\subsection{EDDBM \label{EDDBM}}
We generate the probabilities as following. First, we generate the probabilities on the basis of distance relation given in Equation~\ref{eq15}. Each node $i$ at level $d$ in the $BFT_v$ will get following probability values:
\[p^d= \frac{(\lambda)^{-d}}{\sum_{j\in V \setminus \{v\}} (\lambda)^{-d(j,v)}}.\]
Let $V_d$ be the set of nodes at level $d$ in the $BFT_v$ and $|V_d|$ denotes the number of nodes in set $V_d$. Then to resolve the problem stated in Observation 1 to best extent, at each level $d$, we further tweak the formula on the basis of Equation~\ref{eq16} and get the assigned probability to node $i$ at $d^{th}$ level as:
\begin{equation}
 p_{i}= \frac{p^d |V_d|\cdot deg(i)^{-1}}{\sum_{j\in V_d} deg(j)^{-1} }.
 \label{eddbm}
\end{equation}
\section{Betweenness-Ordering Heuristic} In this section, we discuss our approach for solving the betweenness-ordering problem based on the new non-uniform based sampling EDDBM. Then, we discuss betweenness-ordering problem on $k$ nodes, which we refer to as the $k$-betweenness-ordering.
\subsection{Betweenness-Ordering : Ordering 2 nodes}
Given a graph $G$, this algorithm orders two nodes $u$ and $v$ by first efficiently  estimating their betweenness scores and then comparing the scores. The algorithm is summarized as Algorithm~\ref{al2}.    
\begin{algorithm}[h]$Betweenness\_Ordering(G,u,v,T)$
\begin{algorithmic}[1]
\STATE \textbf{Input.} Graph $G$, node $u$ and node $v$, number of samples $T$.
\STATE Generate $P_u$, set of probabilities for each node based on EDDBM model in $BFT_u$.
\STATE Estimate the betweenness score of node $u$, $B'(u)=Estimate(G,P_u,u,T)$.
\STATE Generate $P_v$= set of probabilities for each node based on EDDBM model in $BFT_v$.
\STATE Estimate the betweenness score of node $v$, $B'(v)=Estimate(G,P_v,v,T)$.

\STATE \textbf{Return.} The result of comparison between $B'(u)$ and $B'(v)$.
\end{algorithmic}
\caption{: Betweenness-Ordering algorithm.}
\label{al2}
\end{algorithm}

Step 2 and Step 3 estimates betweenness score of node $u$. Step 4 and step 5 estimates the betweenness score of node $v$. Step 2 (step 4) generates non-uniform sampling probabilities using the EDDBM model (Equation~\ref{eddbm}) in relation to node $u$ ($v$). Step 3 (step 5) estimates betweenness score of node $u$ ($v$) by passing the generated probabilities to Algorithm~\ref{BO_algo1} .

\subsection{$k$-Betweenness-Ordering : Ordering $k$ nodes}
Algorithm~\ref{al2} orders $k=2$ nodes in a given graph based on the betweenness of nodes. Given a graph $G$ and the problem to order $k$ nodes, we can extend Algorithm~\ref{al2} to handle the case simply by running step 2 and step 3 of Algorithm~\ref{al2} as a sub-procedure for each of the $k$ nodes to estimate their betweenness scores. Once the betweenness centrality of each of the $k$ nodes is estimated, ordering can be done by running a sorting algorithm. One of the most popular sorting algorithms is Merge sorting \cite{Knuth:1998} which takes $O(k \log k)$ time to sort $k$ objects based on object's value.

\subsection{Computation Time}
The time complexity of procedure $Estimate(G,P,v,T)$ (Algorithm~\ref{BO_algo1}) is $O(Tm)$ \cite{Chehreghani:2014}, where $m$ is the number of edges in the graph $G$. It is due to $T$ iterations of breadth first traversals, each of which takes $O(m)$ time. Non-uniform probabilities generation based on EDDBM can be done in $O(m)$ time because it uses one iteration of breadth first traversal. Thus, Algorithm~\ref{al2} ($Betweenness-Ordering(G,u,v,T)$) takes $2\cdot O(m)+2\cdot O(Tm)+1$ = $O(Tm)$ time. Similarly,  the extended version of Algorithm~\ref{al2} ($k-Betweenness-Ordering(G,U)$) takes $O(kTm+k\log k)$ time, where $O(kTM)$ factor is due to k times call to Algorithm~\ref{BO_algo1} and $O(k\log k)$ factor is for sorting. In section \ref{samples_r}, we note that only a constant number of iterations ($T$) seems heuristically to be enough for providing efficient betweenness-ordering. Thus, the running time of the proposed heuristic $Betweenness-Ordering(G,u,v, T)$ is $O(m)$ in this paper by fixing $T$ as a constant. Similarly, if $k$ is very smaller than $n$ (number of total nodes) in the $k$-betweenness-ordering problem and is a constant then the running time of the proposed heuristic for $k$-Betweenness-Ordering problem also becomes $O(m)$.

\section{Experimental Results \label{BO_5}}
In this section, we discuss the experimental results on datasets including real-world graphs and synthetic graphs. We have implemented all the algorithms in C++. All the simulations were performed on a CentOS 6.5 machine with 2x (Xeon E5-2670V2(10 Core,2.5Ghz)) processors and 96 GB RAM.

\subsection{Dataset} 
\subsubsection{Real Networks \label{net}}
We have picked some real-world networks that are popularly used as benchmark networks for betweenness computation and estimation \cite{Chehreghani:2014,Green12,min12,Riondato,Kas13}. We restricted the detailed analysis only to the networks with the number of nodes less than 40,000 due to computational constraints. We provide a brief summary of these networks in Table \ref{tab_real} and \cite{snap, sparse} can be referred for a detailed description of the networks. We have considered collaboration networks, citation networks, communication network, social network, internet peer to peer network and some other. The columns of the Table \ref{tab_real} consist names of the network instances, the number of nodes ($n$), the average degree of the nodes in the networks (Avg. Deg.), number of nodes with zero betweenness score (Z-BC) and the network type respectively.   

\begin{table}[htbp]
\centering
\caption{Considered Real-World Networks}
\resizebox{.7\columnwidth}{!}{
\begin{tabular}{|l|c|c|c|l|}
\hline
\multicolumn{1}{|c|}{\textbf{Instance name}} & \textbf{$n$} & \textbf{Avg. Deg.} & \textbf{Z-BC } & \multicolumn{1}{c|}{\textbf{Network Type}} \\ \hline \hline
as20000102 \cite{snap}& 6474 & 3.88384 & 3682 & Autonomous systems graph\\ 
Wiki-Vote \cite{snap}& 7115 & 28.3238 & 2517 & Social Network\\ 
wb-cs-stanford \cite{sparse}& 9435 & 5.81388 & 2814 & Web Graph\\ 
CA-HepTh \cite{snap}& 9877 & 5.25929 & 5291 & Collaboration network\\ 
oregon1\_010331 \cite{snap}& 10670 & 4.12409 & 6285 & Autonomous systems graph\\ PGPgiantcompo \cite{sparse}& 10680 & 4.55356 & 5663 & Social Network\\ 
oregon1\_010526 \cite{snap}& 11174 & 4.18991 & 6520 & Autonomous systems graph\\
CA-HepPh \cite{snap}& 12008 & 19.735 & 6304 & Collaboration network\\ 
CA-AstroPh \cite{snap}& 18772 & 21.1006 & 8446 & Collaboration network\\ 
p2p-Gnutella25 \cite{snap}& 22687 & 4.82259 & 9348 & Internet peer-to-peer network\\ 
as-22july06 \cite{sparse}& 22963 & 4.21861 & 11927 & Internet Routers Network\\ 
CA-CondMat \cite{snap}& 23133 & 8.07842 & 12635 & Collaboration network\\ 
Cit-HepTh \cite{snap}& 27770 & 25.3716 & 2345 & Citation Network\\ 
Cit-HepPh \cite{snap}& 34546 & 24.3662 & 2120 & Citation Network\\ 
p2p-Gnutella30 \cite{snap}& 36682 & 4.81588 & 16531 & Internet peer-to-peer network\\ 
Email-Enron \cite{snap}& 36692 & 10.0202 & 23710 & Communication Network\\ \hline
\end{tabular}}
\label{tab_real}
\end{table}

\subsubsection{Synthetic Networks \label{syn_n}}
We considered following types of synthetic graphs:
\begin{enumerate}
\item \textbf{Random Graphs (ER).} For generating random graphs, we have considered the most extensively used random graph generation $G(n,p)$ model given by Erdos Renyi~\cite{ER}. The model takes as input the number of nodes $n$ and a probability $p$. Then for each possible pair of different nodes, it puts an edge with a probability of $p$ and outputs the generated graph. We also referred this probability as edge existential probability in this paper.  

\item \textbf{Scale-free Random Graphs (BA).} For generating scale-free random graphs, we have considered the Barabasi-Albert graph generation model \cite{BA}. Throughout the paper, we denote it as $H(n,k)$. It takes as input the number of nodes ($n$) and an integer $k$. It starts with a complete graph of size $k$ and keep adding random $k$ different edges from new coming nodes to the existing nodes with probability as the normalized degree of old nodes. This model is also referred as preferential attachment model.  
\end{enumerate} 
Table \ref{tab_synthetic} summarizes the details of the considered synthetic graphs where ER\_n\_x stands for $G(n, \frac{n^\frac{1}{x}}{n})$ and BA\_n\_x stands for $H(n, \lfloor \frac{n^\frac{1}{x}}{2}\rfloor)$.  The columns of the Table \ref{tab_synthetic} consist name-label to the networks, size of networks ($n$), other parameters (p/k) that needs to be fixed in the generation of synthetic networks, average degree of the nodes in the networks (Avg. Deg.), edges in the network and average number of nodes with zero betweenness score (Avg. Z-BC) respectively.   
\begin{table}[htbp]
\centering
\tiny
\caption{Considered Synthetic Networks}
\resizebox{.8\columnwidth}{!}{
\begin{tabular}{|c|c|l|l|l|c|}
\hline
\textbf{Instance name} & \textbf{n} & \multicolumn{1}{c|}{\textbf{p/k}} & \multicolumn{1}{c|}{\textbf{Avg. Deg.}} & \multicolumn{1}{c|}{\textbf{Edge}} & \multicolumn{1}{c|}{\textbf{Avg. Z-BC}}\\ \hline \hline
ER\_1k\_2 & 1000 & 0.03162278 & 31.6976 & 15848 &0 \\ 
ER\_1k\_3 & 1000 & 0.01 & 10.01 & 5005 &0.2 \\ 
ER\_1k\_4 & 1000 & 0.00562341 & 5.616 & 2808 & 20.4\\ 
ER\_1k\_8 & 1000 & 0.00237137 & 2.4044 & 1202 & 306.4\\ %\hline
%\textbf{} &  &  &  &  \\ 
ER\_10k\_2 & 10000 & 0.01 & 99.96312 & 499816 & 0\\ 
ER\_10k\_3 & 10000 & 0.002154435 & 21.48508 & 107425 & 0\\ 
ER\_10k\_4 & 10000 & 0.001 & 10.00588 & 50029 & 5.2\\ 
ER\_10k\_8 & 10000 & 0.000316228 & 3.1786 & 15893 & 1728.4
\\ %\hline
%\textbf{} &  &  &  &  \\ 
BA\_1k\_2 & 1000 & 16 & 31.488 & 15744 & 0\\ 
BA\_1k\_3 & 1000 & 5 & 9.95 & 4975 & 0\\ 
BA\_1k\_4 & 1000 & 3 & 5.982 & 2991& 0 \\ 
BA\_1k\_8 & 1000 & 2 & 3.992 & 1996 & 21.2\\ %\hline
%\textbf{} &  &  &  &  \\ 
BA\_10k\_2 & 10000 & 50 & 99.5 & 497500 & 0 \\ 
BA\_10k\_3 & 10000 & 11 & 21.9758 & 109879 & 0\\ 
BA\_10k\_4 & 10000 & 5 & 9.995 & 49975 & 0\\ 
BA\_10k\_8 & 10000 & 2 & 3.9992 & 19996 & 32.8\\ \hline
\end{tabular}}
\label{tab_synthetic}
\end{table}

\subsection{Performance Measurement Tools}
In this section, we discuss various measures used to evaluate the performance of our model. 

\subsubsection{For Betweenness Estimation : Error and Average Error \label{averr}}
Let a graph $G=(V,E)$ with $|V|=n$ is given. Let $BC^e(v)$ be the exact betweenness score of node $v$ in the given graph. Let $BC^a(v)$ be the betweenness score of the same node $v$ computed by Algorithm~\ref{BO_algo1} using probabilities generated by EDDBM. Then, as defined by Chehreghani~\cite{Chehreghani:2014}, the error in computation of betweenness score on node $v$ is computed as : 
\[Er(v) = \frac{|BC^e(v)-BC^a(v)|}{BC^e(v)} \times 100\] 
We define \textit{average error} $E$ in the computation of betweenness score of a set of nodes $U$, $U\subseteq V$, over a graph $G$ as 
\[E=\frac{\sum_{i\in U} Er(i)}{|U|}\]
where $|U|$ denotes the number of nodes in set $U$. To compute average error in the computation of betweenness score in a graph, we considered $U=\{v:v\in V \; and\; BC^e(v)>0\}$ throughout the paper. To find the average error in the betweenness computation for a node, we take mean of the error over five iterations. For synthetic graphs, we take mean of the average error over five such synthetic graphs. \textit{Number of iterations} used for computation of betweenness score is referred as the number of sampled nodes. We denote it by $T$.

\subsubsection{For Betweenness-Ordering : Efficiency and Relaxed Efficiency \label{eff}} 
Let $n$ be the number of nodes in the considered graph. Then, $\binom{n}{2}$ different pairs of nodes are possible. Let $b_{ij}=1$ if the result of betweenness comparison between node $i$ and node $j$ by our algorithm is correct, otherwise $b_{ij}=0$. The efficiency of algorithm for betweenness-ordering of two nodes can be given as 
\[\xi=\frac{\sum_{i=1}^{n-1} \sum_{j=i+1}^n b_{ij}}{\binom{n}{2}}.\]
In real world scenarios when two nodes possess very close betweenness ranks, error in the betweenness-ordering of those two nodes does not matter much. For example importance of the top central node or second top central node is very close. Thus, a relaxed version of the efficiency measure can be modeled. We take a threshold $t$ and relax the ordering of nodes if the difference between the betweenness ranks of both the nodes is less or equal to $t$. By relaxing, we mean that we do not consider those pairs for measuring the efficiency of an algorithm.
Let $P_t$ be the set of all pairs of nodes with betweenness rank difference greater than $t$ and $|P_t|$ denotes the cardinality of set $P_t$. Let $b_{ij}$ be a flag variable which gets value $1$ for correct comparison and $0$ otherwise. Then we can redefine the relaxed efficiency as:
\[\xi^t=\frac{\sum_{(i,j)\in P_t} b_{ij}}{|P_t|}.\]
At $t=0$, $\xi^t=\xi$. In \ref{relaxed_ef}, we show results for $t=\{2,3,5,10\}$ on considered synthetic graphs. Next, we name the algorithms picked for betweenness estimation and ordering analysis.
\subsection{Considered Competitive Algorithms \label{comalg}}
In this section, we mention the algorithms picked for comparative analysis with our approaches for betweenness-ordering and betweenness estimation. We consider following labels \textit{BP/B, LS, MC, BOLT, 2-BC }for the Brandes and Pich's (time bounded version of Bader et al.'s~\cite{Bader:2007}) uniform sampling based approximation algorithm, Geisberger et al.'s linear scaling based algorithm \cite{Geisberger:2008}, a recent algorithm by Chehreghani~\cite{Chehreghani:2014}, our algorithm and  $k$-betweenness algorithm due to Gkorou et al.~\cite{Gkorou} with $k=2$ respectively. The first four algorithms are the node sampling based (probabilistic) algorithms, and we fix an equal number of samples for each one of these algorithms. The last, 2-BC algorithm is a deterministic algorithm and takes a lot more time than the sampling based algorithms. We have considered this algorithm to show that even with very few samples (very less time), most of the times, our algorithm outperforms this deterministic algorithm. A recent path sampling based  Riondato and Kornaropoulos's~\cite{Riondato} randomized algorithm is theoretically sound, but in the small time frame fixed by us, it does not perform well for estimation or ordering on any considered network. Thus, we skip the results by their algorithm. Next, we see various plots to analyze our model and algorithms.

\subsection{Plots for Betweenness estimation and ordering}
In this section, we evaluate the performance of EDDBM (from section \ref{EDDBM}) for estimation and ordering through various plots. Next, with the help of different plots on synthetic and real-world networks, we compare the accuracy of EDDBM in comparison to DBM (from section \ref{DBM}).

\subsubsection{Comparison of probabilities assigned by DBM, EDDBM, and optimal model}
Plots in this section analyze the performance of EDDBM experimentally and shows that EDDBM generates probabilities very close to the optimal probabilities. These plots also compare EDDBM with DBM.
We draw the plots in Figure~\ref{figps} for four different synthetic networks which we pick from Table \ref{tab_synthetic}. 

\begin{figure}[htb]
\centering
\subfigure[BA\_1k\_4]{%
\includegraphics[width=.48\columnwidth]{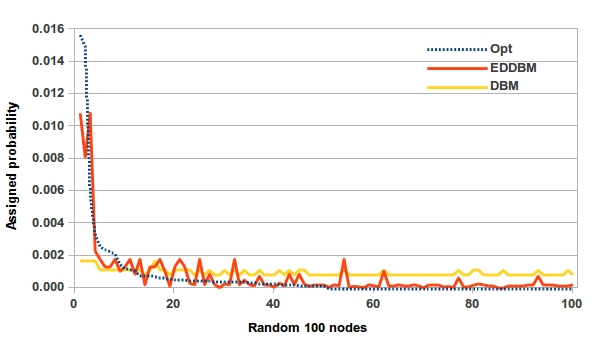}
\label{figps.1}}
\subfigure[BA\_1k\_3]{%
\includegraphics[width=.48\columnwidth]{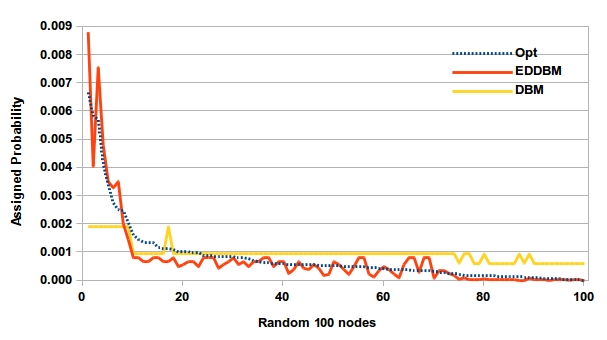}
\label{figps.2}}
\subfigure[ER\_1k\_4]{%
\includegraphics[width=.48\columnwidth]{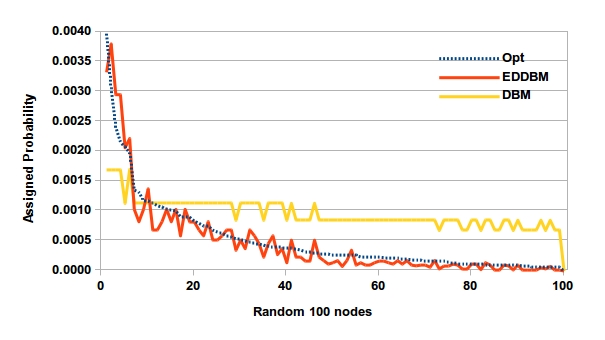}
\label{figps.3}}
\subfigure[ER\_1k\_3]{%
\includegraphics[width=.48\columnwidth]{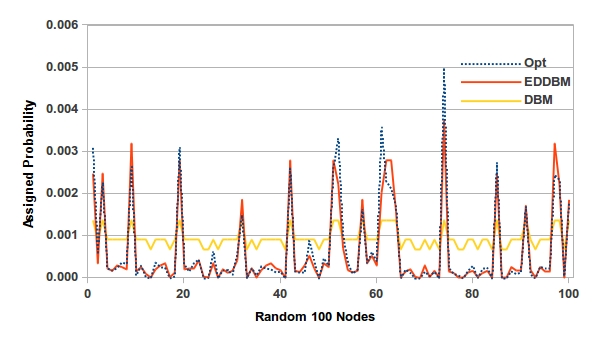}
\label{figps.4}}
\caption{Comparison of probabilities assigned by DBM and EDDBM vs the optimal model (opt) in four different synthetic networks.}
\label{figps}
\end{figure}

For generating the plots, we picked an arbitrary node from each network. With the assumption to estimate betweenness of this node, we assign probabilities to all the nodes in the network using EDDBM model, optimal sampling model (Opt), and DBM. Each network consists a large number of nodes. To draw a clear plot, we randomly picked 100 nodes and plotted probabilities assigned by the sampling models mentioned above for only these 100 randomly selected nodes. The x-axis represents the 100 chosen nodes, and the y-axis represents the probabilities assigned by DBM, EDDBM, and the optimal model (Opt). In the first three plots in Figure~\ref{figps}, we sorted the randomly picked 100 nodes in descending order based on the optimal probabilities assigned to them before plotting. The last one plot in Figure~\ref{figps.4} is without such sorting process.

In the Figure~\ref{figps}, it is easy to observe that EDDBM is much better than DBM. EDDBM generates probabilities very close to the optimal probabilities. In the Figure~\ref{figps.4}, we note that the plot of probabilities by EDDBM achieves similar characteristic peaks as the plot of optimal probabilities get. The analysis of plots infers that, unlike DBM, EDDBM identifies the nodes with high contribution and assigns significant probabilities to them. EDDBM also focuses on the nodes contributing very less and tries to assign smaller probabilities to them which was not well handled by DBM.

\subsubsection{Average Error and Efficiency vs Size of Graphs ($n$) \label{errdec}}
In this section, we plot the average error in the computation of betweenness score and the average efficiency in ordering the nodes based on the betweenness scores in a graph in respect of the order (number of nodes) of graphs. We generated graphs with $n=100$ to $n=1000$ with a step of $100$. For each $n$, we generated $5$ graphs and averaged the average error and efficiency over all the $5$ graphs. Figure~\ref{p-size} contains the plots. The plot in Figure~\ref{p-size-e} depicts the change in the mean error and plot in Figure~\ref{p-size-o} represents the change in the mean efficiency while changing the size of graphs, but keeping average degree as a constant. We plotted the results for ER graphs with average degree = \{3, 5, 10\} and for BA graphs with average degree = \{4, 6, 10\}. From the plots, we can infer that the mean error in the computation of betweenness score decreases and the mean efficiency in ordering the nodes based on betweenness score increases with an increase in the size of graphs.

\begin{figure}[H]
\centering
\subfigure[Average Error vs Size of Graphs]{%
\includegraphics[width=.75\textwidth]{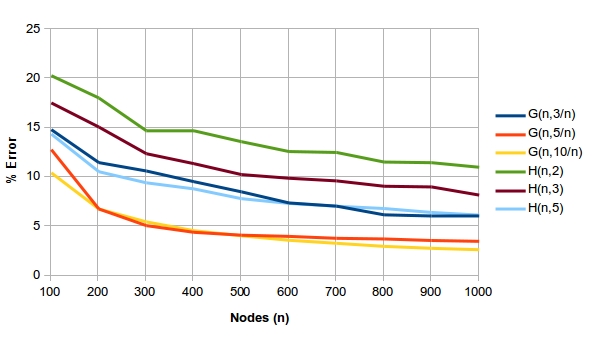}
\label{p-size-e}}
\subfigure[Efficiency vs Size of Graphs]{%
\includegraphics[width=.75\textwidth]{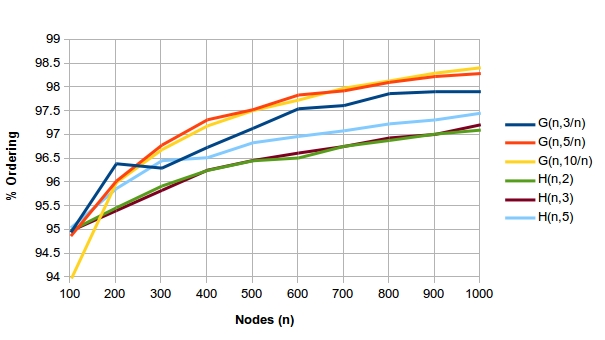}
\label{p-size-o}}
\caption{Average Error and Efficiency vs Size of Synthetic (ER and BA) Graphs}
\label{p-size}
\end{figure}

\subsubsection{Average Error and Efficiency vs Number of Sampled Nodes ($T$)\label{samples_r}}
In this section, we plot the average error in the computation of betweenness centrality using EDDBM and efficiency in ordering the nodes based on betweenness score using Algorithm~\ref{al2}, when the number of sampled nodes (no of iterations) were $T=X$. These plots were drawn to inspect what value of $T$ that can suffice for a good result, i.e., betweenness estimation with less error and betweenness-ordering with high efficiency. Figure~\ref{iter-real} is the plots of change in the average error versus $T$ and change in the average efficiency versus $T$ on the considered real-world graphs. Plots on the considered synthetic graphs are available in the appendix due to the limitation on the number of figures. 
\begin{figure}[H]
\centering
\subfigure[Average Error]{%
\includegraphics[width=.75\columnwidth]{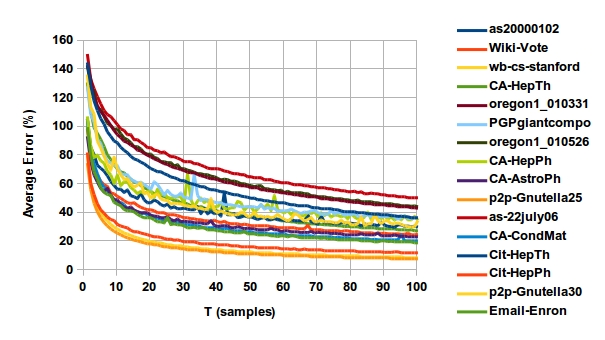}
\label{iter-order-er}}
\subfigure[Average Efficiency]{%
\includegraphics[width=.75\columnwidth]{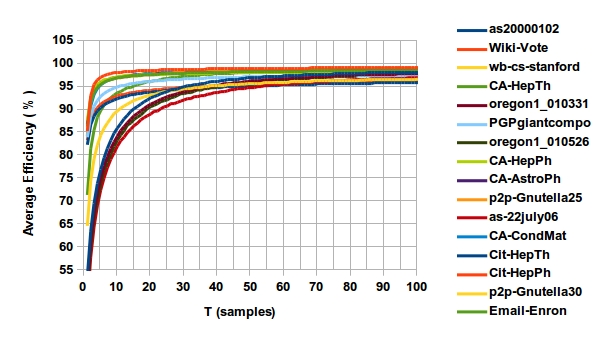}
\label{iter-order-er}}
\caption{Iterative performance of our approach on considered real-world networks.}
\label{iter-real}
\end{figure}

We start all the plots in this section from $X=1$. In the plots in this section, the average error reduces and average efficiency increases very sharply when $X$ varies from $1$ to $15$ or $25$. After $X=25$ there is a very small change in the average error and efficiency. Due to our focus on a quick betweenness-ordering heuristic, we concentrate on the iterative mean efficiency(ordering) performance plot on real-world networks given in Figure~\ref{iter-real}. It is notable that at $T$ = 25, on all considered real-world graphs, the average efficiency reaches beyond 90\% and for most of them even beyond 95\%. Thus, we set \textbf{$T$ = 25} to achieve experimental results in this paper. By reducing $T$ to a constant, one can suspect that the error might increase in larger graphs, but in big networks, our model performs much better. The reason is in the previous section that the average error decreases and efficiency increases with an increase in the number of nodes which neutralizes the effect of the increase in the error by keeping $T$ as constant. Similar results have been observed in the considered synthetic graphs and are available in the appendix.

\begin{figure}[htb]
\centering
\subfigure[On synthetic networks]{%
\includegraphics[width=.7\columnwidth]{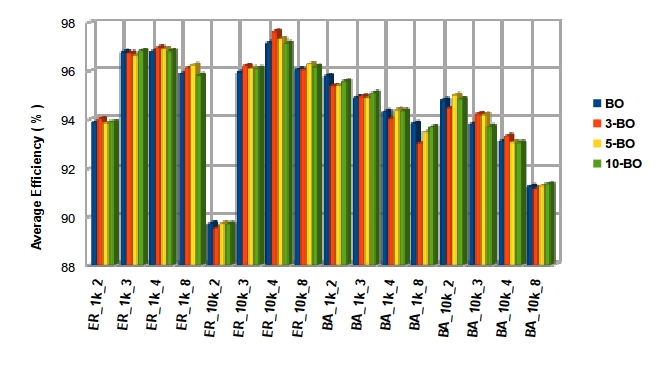}
\label{k-s}}
\;
\subfigure[On real-world networks]{%
\includegraphics[width=.7\columnwidth]{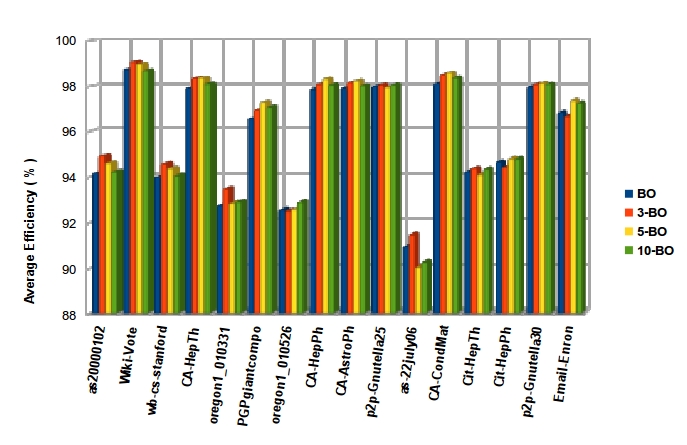}
\label{k-r}}
\caption{Average $k$-betweenness-ordering efficiency for different values of $k$ }
\label{k-fig}
\end{figure} 
\subsubsection{Average efficiency in $k$-betweenness-ordering vs $k$}
In this section, we plot the average efficiency of BOLT for ordering $k$ nodes in respect of the considered synthetic and real-world graphs. We calculate the average efficiency for $k$-betweenness-ordering as following. In a given graph, we randomly pick a set of $k$ nodes and calculate the efficiency in ordering these $k$ nodes based on the formula given in section \ref{eff}. We do this for 1000 iterations and take the mean of the efficiency in these 1000 iterations, and call it average efficiency.

The plots in Figure~\ref{k-fig} show the average efficiency of the extension of Algorithm~\ref{al2} for $k$-betweenness-ordering on various synthetic networks (Figure~\ref{k-s}) and real-world networks (Figure~\ref{k-r}). Label BO, denotes the efficiency of betweenness-ordering. Labels 3-BO, 5-BO, and 10-BO denotes the average efficiency of BOLT approach for $k$-betweenness-ordering on $k=3,5,10$ respectively. Plots show that the performance of the extended Algorithm~\ref{al2} is nearly same for $k$-betweenness-ordering as the performance of Algorithm~\ref{al2} for the betweenness-ordering (ordering of two nodes).  

\subsubsection{Correlation in ordering / Average efficiency}
In this section, we plot the average \textit{Spearman rank correlation (rho)} between the results produced by Algorithm~\ref{al2} and the other node-sampling based approaches on synthetic and real-world graphs. Rho is a standard ranking correlation that measures the similarity between the ordering of two ranking algorithms. The plots are in Figure~\ref{cor}.

Correlation of BOLT is very high (very close to 1) for almost all networks. BOLT outperforms all node sampling based algorithms. In very dense networks, 2-BC sometimes produces a better result than BOLT, but it should be noted that the difference in correlations are minute even when 2-BC takes a lot more time than BOLT. In sparse networks, BOLT produces much better results than 2-BC in a smaller amount of time. 

% We have plotted the average efficiency for $n=50$ to $n=200$ with a step of 10. For each $n$, we generated $10$ graphs and then averaged the efficiency over $10$ values. Efficiency exhibits the fraction of correct comparisons over all possible comparisons in a graph and can be computed by the formula given in section \ref{eff}. These plots investigate the average performance of our model for betweenness-ordering of two nodes. We plotted the efficiency of DBM and the uniform sampling model. Our model performed the best out of the considered models. We also plotted the correlation (EDDBM\_C) between the betweenness ranking of vertices by our model and their exact betweenness rank to evaluate the accuracy of ordering. Ordering all nodes by computing betweenness score for each node by Algorithm 1 and then sorting is very expensive than exactly computing the betweenness scores using Brandes's algorithm for all nodes and then sorting. The correlation plot is just to understand the accuracy of Algorithm 1, when used with EDDBM. As we can see in the plots, the efficiency of our model for betweeness-ordering in even the small sized graphs is very high (greater than $0.90$) and keeps increasing with the size of graph. Even for graph with size $200$, it reaches very close to $0.95$ which will further increase in large sized graphs. The correlation between the ordering achieved by our algorithm and the exact ordering is also very high and reached very close to 1 for graphs of size just $200$.

In next section, we will discuss the betweenness estimation and betweenness-ordering results achieved for the considered synthetic networks and some of the real-world networks.

\subsection{Average Error and Efficiency in Graphs} Here, we discuss and compare the results obtained by BOLT and other competitive algorithms that are mentioned in section \ref{comalg} on considered synthetic networks and several real-world networks. %The tables are available in the appendix as supplementary material.  

\begin{figure}[H]
\centering
\subfigure[On synthetic networks]{%
\includegraphics[width=.7\columnwidth]{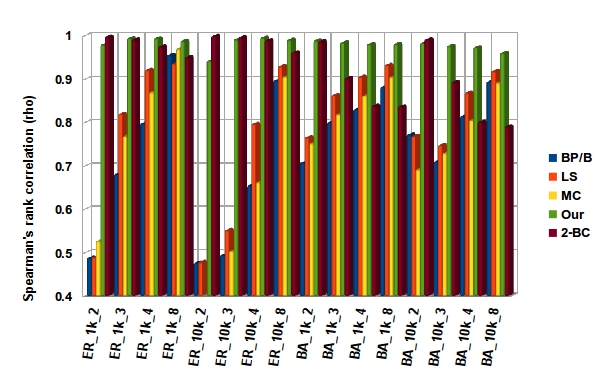}
\label{cor_syn}}
\;
\subfigure[On real-world networks]{%
\includegraphics[width=.7\columnwidth]{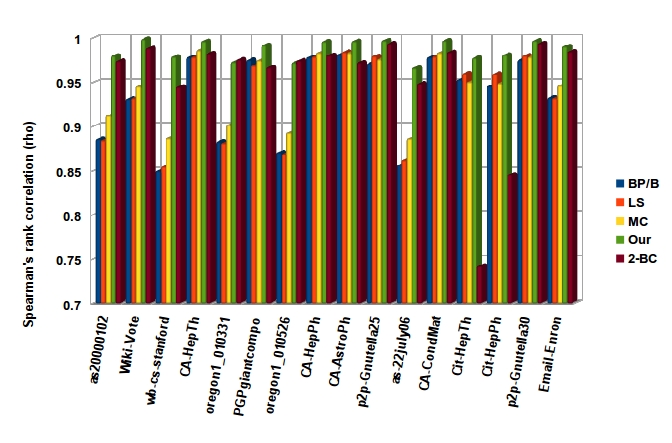}
\label{cor_real}}
\caption{Average spearman rank correlation (rho) of various algorithms in the considered synthetic and real-world networks. }
\label{cor}
\end{figure}
\subsubsection{Average Error and Efficiency in Synthetic Graphs}
In this section, we analyze the results over synthetic graphs mentioned in section \ref{syn_n}. The average of average error is 88.138, 58.063, 67.719, \textbf{12.564}, and 86.548 percentage and the standard deviation in the mean error is 17.235, 16.894, 19.243, \textbf{8.399}, 20.646 by BP/B, LS, MC, \textbf{BOLT}, and 2-BC respectively. The average of average efficiency is 78.913, 81.666, 79.455,    \textbf{94.684}, and    87.383 percentage and the standard deviation in the average efficiency is 7.499, 7.669, 7.945, \textbf{2.046}, 10.344 by BP/B, LS, MC, \textbf{BOLT}, and 2-BC respectively. It is easy to observe that BOLT outperforms all mentioned sampling-based algorithms for both, estimation and ordering by a huge margin. The standard deviation in the performance (average error and efficiency) of BOLT is comparatively smaller than other algorithms. It shows that BOLT's performance is consistent. The 2-BC algorithm performs better for ordering, but not for estimation in very dense graph. It is because of very small average distance between nodes. But, it should be noted that 2-BC takes several folds more time that BOLT. In moderately denser or sparser graphs, BOLT even outperforms 2-BC by a significant margin. Our model performs relatively better in dense graphs than in sparse graphs. 
\begin{table}[htbp]
\centering
\caption{Average error and average efficiency in synthetic networks when $T=25$ }
\resizebox{\columnwidth}{!}
{
\begin{tabular}{|l||c|c|c|c|c||c|c|c|c|c|}
\hline
\multicolumn{1}{|c|}{\textbf{Instance }}&\multicolumn{5}{|c|}{\textbf{Average Error }}&\multicolumn{5}{|c|}{\textbf{Average Efficiency }}\\ \hline
\textbf{Instance } & \textbf{BP/B} & \textbf{LS} & \textbf{MC} & \textbf{BOLT} & \textbf{2-BC}& \textbf{BP/B} & \textbf{LS} & \textbf{MC} & \textbf{BOLT} & \textbf{2-BC}\\ \hline \hline
\textbf{ER\_1k\_2} & 66.291 & 50.603 & 43.433 & \textbf{5.425} & 53.982 & 66.428 & 67.685 & 68.193 & \textbf{93.927} & 98.346 \\ 
\textbf{ER\_1k\_3} & 70.015 & 37.699 & 44.585 & \textbf{5.345} & 95.880 & 74.186 & 80.938 & 78.054 & \textbf{96.828} & 95.481 \\ 
\textbf{ER\_1k\_4} & 67.505 & 42.140 & 44.097 & \textbf{7.505} & 99.033  & 80.398 & 88.441 & 83.943 & \textbf{96.854} & 90.025\\ 
\textbf{ER\_1k\_5} & 63.252 & 62.835 & 37.646 & \textbf{18.257} & 99.411 & 90.980 & 89.678 & 92.646 & \textbf{95.954} & 83.221\\% \hline
\textbf{ER\_10k\_2} & 113.418 & 90.286 & 82.129 & \textbf{5.281} & 53.826  & 67.173 & 67.310 & 63.935 & \textbf{89.760} & 99.295\\ 
\textbf{ER\_10k\_3} & 89.656 & 52.754 & 75.481 & \textbf{4.381} & 98.045  & 67.504 & 69.518 & 67.817 & \textbf{96.023} & 97.256\\ 
\textbf{ER\_10k\_4} & 92.911 & 47.576 & 72.433 & \textbf{4.793} & 99.699 & 73.776 & 79.966 & 74.018 & \textbf{97.215} & 92.901 \\ 
\textbf{ER\_10k\_5} & 83.714 & 65.528 & 66.032 & \textbf{14.084} & 99.933  & 87.716 & 88.738 & 87.845 & \textbf{96.128} & 82.815\\% \hline
\textbf{BA\_1k\_2} & 83.915 & 64.922 & 56.750 & \textbf{7.987} & 52.232 & 77.803 & 79.644 & 79.024 & \textbf{95.855} & 95.583 \\ 
\textbf{BA\_1k\_3} & 82.078 & 47.664 & 61.647 & \textbf{12.487} & 93.155 & 80.854 & 84.555 & 81.845 & \textbf{94.965} & 86.927\\ 
\textbf{BA\_1k\_4} & 84.845 & 45.439 & 63.220 & \textbf{17.155} & 97.567 & 82.889 & 87.979 & 84.245 & \textbf{94.372} & 79.543\\ 
\textbf{BA\_1k\_8} & 82.732 & 48.137 & 66.511 & \textbf{23.427} & 98.868 & 86.115 & 89.199 & 86.731 & \textbf{93.889} & 70.505\\ %\hline
\textbf{BA\_10k\_2} & 123.117 & 100.389 & 99.124 & \textbf{8.208} & 48.404 & 80.993 & 81.583 & 76.898 & \textbf{94.873} & 96.455\\ 
\textbf{BA\_10k\_3} & 100.749 & 65.610 & 90.002 & \textbf{13.017} & 95.667  & 77.458 & 78.496 & 77.827 & \textbf{93.859} & 86.472 \\ 
\textbf{BA\_10k\_4} & 103.962 & 54.955 & 89.731 & \textbf{18.844} & 99.179 & 81.785 & 84.967 & 81.571 & \textbf{93.161} & 79.137\\ 
\textbf{BA\_10k\_8} & 102.052 & 52.474 & 90.697 & \textbf{34.829} & 99.892 & 86.549 & 87.962 & 86.687 & \textbf{91.289} & 64.161\\ \hline
\textbf{Average Error/Efficiency}&88.138 & 58.063	&67.720	&\textbf{12.564} & 86.548 & 78.913 & 81.666 & 79.455 & \textbf{94.685} & 87.383\\ \hline
\textbf{Standard Deviation}&17.235	&16.894	&19.243	&\textbf{8.399} & 20.646 & 7.499 & 7.669 & 7.945 & \textbf{2.046} & 10.344\\ \hline

\end{tabular}}
\label{r-tab1}
\end{table}

The results are summarized in Table \ref{r-tab1}. The first column in the table consists label of the network instance. The next five columns in Table \ref{r-tab1} are the estimation performance, average error in estimating betweenness, by various algorithms. The last five columns are the ordering performance results, average efficiency in ordering two nodes based on betweenness score, by different algorithms.

\subsubsection{Average Error and Efficiency in Real-world Graphs \label{real} }
This section presents and discusses the simulation results on various real networks considered in section \ref{net}. After extracting the networks, we converted the networks into unweighted undirected networks, if required. Then we removed multi-edges, self-loops, and isolated nodes if existing. The average of average error is 109.469,    69.210,    98.603,    \textbf{47.063}, and 95.289 percentage and the standard deviation in average error is 13.845, 18.412, 12.606, \textbf{19.843}, 8.859 by BP/B, LS, MC, \textbf{BOLT}, and 2-BC respectively. The average of average efficiency is 84.803, 84.841, 87.294, \textbf{95.861}, 92.317 percentage and the standard deviation in average efficiency is 10.095, 10.150, 8.065, \textbf{2.465}, 4.747 by BP/B, LS, MC, \textbf{BOLT}, and 2-BC respectively. BOLT again superseded all other considered algorithms for betweenness estimation. Although the standard deviation of the error for estimating betweenness score is higher than other algorithms, still the average error is relatively less. The large standard deviation is due to different nature, structure and size of graphs. But, the standard deviation in the efficiency for ordering nodes using BOLT is very less which shows a consistently good performance of BOLT in ordering nodes based on betweenness score. The estimation results can be improved by increasing the value of $T$ that is considered 25 for all computations.

\begin{table}[htbp]
\centering
\caption{Average error and average efficiency in real-world networks when $T=25$ }
\resizebox{\columnwidth}{!}
{
\begin{tabular}{|l||c|c|c|c|c||c|c|c|c|c|}
\hline
\multicolumn{1}{|c|}{\textbf{Instance }}&\multicolumn{5}{|c|}{\textbf{Average Error }}&\multicolumn{5}{|c|}{\textbf{Average Efficiency }}\\ \hline
\textbf{Instance } & \textbf{BP/B} & \textbf{LS} & \textbf{MC} & \textbf{BOLT} & \textbf{2-BC}& \textbf{BP/B} & \textbf{LS} & \textbf{MC} & \textbf{BOLT} & \textbf{2-BC}\\ \hline \hline
\textbf{as20000102} & 110.466 & 70.932 & 100.000 & \textbf{66.265} & 98.683 & 76.398 & 76.336 & 80.574 & \textbf{94.224} & 93.390\\ 
\textbf{Wiki-Vote} & 117.143 & 78.948 & 104.930 & \textbf{22.355} & 93.650 & 81.246 & 81.296 & 84.614 & \textbf{98.779} & 96.037  \\
\textbf{wb-cs-stanford} & 135.742 & 126.957 & 123.029 & \textbf{52.811} & 62.970 & 63.901 & 63.685 & 71.806 & \textbf{94.041} & 91.597 \\
\textbf{CA-HepTh} & 98.828 & 57.641 & 77.841 & \textbf{35.233} & 96.075  & 93.160 & 92.910 & 95.442 & \textbf{97.951} & 94.732 \\ 
\textbf{oregon1\_010331} & 134.948 & 76.411 & 109.810 & \textbf{74.394} & 98.711 & 75.061 & 74.862 & 78.706 & \textbf{92.793} & 93.659 \\ 
\textbf{PGPgiantcompo} & 96.975 & 52.619 & 93.898 & \textbf{54.883} & 96.212 & 93.092 & 93.064 & 94.650 & \textbf{97.886} & 94.629\\ 
\textbf{oregon1\_010526} & 115.421 & 77.324 & 115.177 & \textbf{74.300} & 98.722 & 72.859 & 72.453 & 77.153 & \textbf{92.594} & 93.461 \\ 
\textbf{CA-HepPh} & 95.320 & 55.769 & 86.426 & \textbf{50.714} & 96.254 & 93.092 & 93.064 & 94.650 & \textbf{97.886} & 94.629 \\ 
\textbf{CA-AstroPh} & 98.069 & 54.855 & 91.558 & \textbf{36.938} & 97.947 & 93.778 & 94.178 & 94.496 & \textbf{97.934} & 93.824\\ 
\textbf{p2p-Gnutella25} & 102.861 & 63.462 & 83.838 & \textbf{16.760} & 99.962 & 92.922 & 93.918 & 93.741 & \textbf{98.019} & 94.969\\ 
\textbf{as-22july06} & 131.187 & 82.954 & 115.373 & \textbf{80.428} & 98.965 & 70.131 & 70.965 & 75.438 & \textbf{90.969} & 90.567\\ 
\textbf{CA-CondMat} & 95.170 & 53.324 & 89.418 & \textbf{34.532} & 96.333 & 93.234 & 93.131 & 94.671 & \textbf{98.147} & 95.144 \\
\textbf{Cit-HepTh} & 104.839 & 59.056 & 96.572 & \textbf{44.896} & 98.253 & 91.014 & 91.204 & 91.043 & \textbf{94.270} & 78.306\\ 
\textbf{Cit-HepPh} & 109.118 & 58.523 & 95.126 & \textbf{39.311} & 98.949 & 90.171 & 91.209 & 90.744 & \textbf{94.695} & 83.741\\ 
\textbf{p2p-Gnutella30} & 100.685 & 64.435 & 90.594 & \textbf{18.407} & 99.957 & 93.428 & 93.980 & 94.120 & \textbf{97.989} & 95.347\\ 
\textbf{Email-Enron} & 104.747 & 74.159 & 104.052 & \textbf{50.781} & 92.997 & 84.086 & 83.866 & 87.182 & \textbf{96.876} & 95.275\\ \hline
\textbf{Average Error/Efficiency}&109.470	&69.211	&98.603	&\textbf{47.063}&	95.290 &84.803	&84.841	&87.295	&\textbf{95.861}	&92.317\\ \hline
\textbf{Standard Deviation}&13.845 &	18.412 &	12.606 &	\textbf{19.843}&	8.859 &10.095	&10.150	&8.065	&\textbf{2.465}	&4.747\\ \hline

\end{tabular}}
\label{r-tab3}
\end{table}

For evaluating the more general performance of BOLT, in addition to the 16 networks from section \ref{net}, we have considered 54 more networks that cover most of the different networks available at \cite{snap,vlado} of size (100, 100k). The betweenness-ordering results for $T=25$ and $T=50$ are calculated on all the networks mentioned in section \ref{net} and the other 54 picked networks. We achieve average efficiency 96.586 and 97.632 percentage with a standard deviation of 2.384 and 1.447 percentage for $T$ = 25, 50 respectively ignoring three special networks. %The details of the other 54 networks and ordering results on all real-world networks are provided in the appendix as supplementary material.
This section presents and discusses the simulation results on various real-world networks. After extracting the networks, we converted the networks into unweighted undirected networks, if required. Then we removed multi-edges, self-loops, and isolated nodes if existing. We summarize the obtained results in the Table \ref{r-tab3}. The columns are in similar order as in Table \ref{r-tab1} respectively. Here, Table \ref{r-tab3} compiles the performance of estimation and ordering results.

Further, we present and discuss the ordering results by BOLT on more real-world networks. We compile the obtained results in Table \ref{r-tab5}. The first five columns of Table \ref{r-tab5} contain serial number, the name of the network instances, the size of networks ($n$), the average degree of the nodes (Avg. Deg.), the number of nodes with zero betweenness score (Z-BC) respectively. Next column contains the average efficiency of BOLT for ordering when $T$ = 25 and all $\binom{n}{2}$ pairs are considered for calculating the efficiency. Next column contains the average efficiency of BOLT for ordering when $T$ = 25 and only the pairs that consist at least one node with nonzero betweenness scores, are considered for computing the efficiency. The next two columns are same as the columns 6-7 except the efficiency is calculated when $T$=50 (50 samples) is set in our algorithm. The picked networks almost cover the different networks data-sets available at \cite{snap,vlado} of size (100, 100k). The results show that the efficiency of our algorithm is very high and very close to the exact ordering for only a constant number of samples.

On three of the real-world networks that were picked, the ordering results were not good. The results are present in the last three rows of Table \ref{r-tab5}. The reason for this bad performance is a unique property of these networks. In these networks, most of the nodes share same betweenness score. Our method probabilistically estimates the betweenness score. The defined efficiency formula considers when two nodes share same actual betweenness score, they get same rank and the efficiency only increases if the considered estimation algorithm also assigns both of the nodes exactly same score. But, by any probabilistic algorithm, even if it is very efficient, due to the probabilistic nature, it is very less probable (nearly impossible) that it will be able to assign the same score. Thus, BOLT gets lower efficiency though the average errors in the estimation were small. 
\begin{comment}
\begin{center}
\begin{table}[htbp]
\caption{Average Efficiency (in \%) of BOLT on three special real-world networks setting  $T$ = 25, 50}
\resizebox{1\columnwidth}{!}
{
\begin{tabular}{|c|l|c|c|c|c|c|c|c|}
\hline
\textbf{} & \textbf{Instance name} & \textbf{n} & \textbf{Avg. Deg.} & \textbf{Z-BC} & \textbf{Ordering\_25} & \textbf{Ordering\_25 nz} & \textbf{Ordering\_50} & \textbf{Ordering\_50 nz} \\ \hline
1 & \textbf{GD06\_theory} & 101 & 3.762 & 0 & 20.875 & 20.875 & 20.222 & 20.222 \\ \hline
2 & \textbf{GD96\_b} & 111 & 3.477 & 0 & 84.842 & 84.842 & 85.789 & 85.789 \\ \hline
3 & \textbf{GD98\_c} & 112 & 3.000 & 0 & 62.825 & 62.825 & 66.599 & 66.599 \\ \hline
\end{tabular}}
\label{r-tab6}
\end{table}
\end{center} 
\end{comment}
\begin{center}
\begin{table}[htbp]
\centering
\caption{Average Efficiency (in \%) of BOLT on real-world networks setting  $T$ = 25, 50}
\resizebox{.75\columnwidth}{!}
{
\begin{tabular}{|c|l|c|c|c|c|c|c|c|}
\hline
\textbf{S.N.} & \textbf{Instance name} & \textbf{n} & \textbf{Avg. Deg.} & \textbf{Z-BC} & \textbf{Ordering\_25} & \textbf{Ordering\_25 nz} & \textbf{Ordering\_50} & \textbf{Ordering\_50 nz} \\ \hline
\textbf{1} & \textbf{as20000102} & 6474 & 3.884 & 3682 & 96.092 & 94.224 & 97.949 & 96.968 \\ \hline
\textbf{2} & \textbf{Wiki-Vote} & 7115 & 28.324 & 2517 & 98.932 & 98.779 & 99.174 & 99.056 \\ \hline
\textbf{3} & \textbf{wb-cs-stanford} & 9435 & 5.814 & 2814 & 94.571 & 94.041 & 96.186 & 95.814 \\ \hline
\textbf{4} & \textbf{CA-HepTh} & 9877 & 5.259 & 5291 & 98.539 & 97.951 & 98.856 & 98.396 \\ \hline
\textbf{5} & \textbf{oregon1\_010331} & 10670 & 4.124 & 6285 & 95.293 & 92.793 & 97.643 & 96.390 \\ \hline
\textbf{6} & \textbf{PGPgiantcompo} & 10680 & 4.554 & 5663 & 97.560 & 96.606 & 98.011 & 97.233 \\ \hline
\textbf{7} & \textbf{oregon1\_010526} & 11174 & 4.190 & 6520 & 95.115 & 92.594 & 97.407 & 96.068 \\ \hline
\textbf{8} & \textbf{CA-HepPh} & 12008 & 19.735 & 6304 & 98.468 & 97.886 & 98.672 & 98.166 \\ \hline
\textbf{9} & \textbf{CA-AstroPh} & 18772 & 21.101 & 8446 & 98.352 & 97.934 & 98.674 & 98.338 \\ \hline
\textbf{10} & \textbf{p2p-Gnutella25} & 22687 & 4.823 & 9348 & 98.355 & 98.019 & 98.809 & 98.566 \\ \hline
\textbf{11} & \textbf{as-22july06} & 22963 & 4.219 & 11927 & 93.405 & 90.969 & 96.326 & 94.969 \\ \hline
\textbf{12} & \textbf{CA-CondMat} & 23133 & 8.078 & 12635 & 98.700 & 98.147 & 99.000 & 98.575 \\ \hline
\textbf{13} & \textbf{Cit-HepTh} & 27770 & 25.372 & 2345 & 94.311 & 94.270 & 95.286 & 95.253 \\ \hline
\textbf{14} & \textbf{Cit-HepPh} & 34546 & 24.366 & 2120 & 94.715 & 94.695 & 95.647 & 95.631 \\ \hline
\textbf{15} & \textbf{p2p-Gnutella30} & 36682 & 4.816 & 16531 & 98.397 & 97.989 & 98.831 & 98.533 \\ \hline
\textbf{16} & \textbf{Email-Enron} & 36692 & 10.020 & 23710 & 98.181 & 96.876 & 98.998 & 98.280 \\ \hline
\textbf{17} & \textbf{as19990829} & 103 & 4.641 & 43 & 98.835 & 98.593 & 99.101 & 98.915 \\ \hline
\textbf{18} & \textbf{facebook\_combined} & 4039 & 43.691 & 342 & 96.396 & 96.370 & 97.242 & 97.222 \\ \hline
\textbf{19} & \textbf{CA-GrQcNew} & 5242 & 5.526 & 3236 & 98.957 & 98.315 & 99.143 & 98.616 \\ \hline
\textbf{20} & \textbf{P2p-Gnutella04} & 10876 & 7.355 & 2484 & 97.634 & 97.504 & 98.280 & 98.186 \\ \hline
\textbf{21} & \textbf{oregon2\_010331} & 10900 & 5.721 & 6096 & 96.706 & 95.207 & 98.311 & 97.543 \\ \hline
\textbf{22} & \textbf{Oregon2\_010526} & 11461 & 5.712 & 6290 & 96.345 & 94.770 & 98.042 & 97.199 \\ \hline
\textbf{23} & \textbf{P2p-Gnutella24} & 26518 & 4.930 & 11014 & 98.214 & 97.842 & 98.706 & 98.436 \\ \hline
\textbf{24} & \textbf{P2p-Gnutella31} & 62586 & 4.726 & 28829 & 98.241 & 97.768 & 98.708 & 98.360 \\ \hline
\textbf{25} & \textbf{Soc-Epinions1} & 75879 & 10.694 & 41048 & 98.006 & 97.181 & 98.529 & 97.921 \\ \hline
\textbf{26} & \textbf{Slashdot0811} & 77360 & 12.130 & 30164 & 97.575 & 97.140 & 98.179 & 97.853 \\ \hline
\textbf{27} & \textbf{Slashdot0902} & 82168 & 12.273 & 30855 & 97.493 & 97.081 & 98.151 & 97.848 \\ \hline
\textbf{28} & \textbf{GD99\_c} & 105 & 2.286 & 36 & 95.190 & 94.563 & 95.769 & 95.217 \\ \hline
\textbf{29} & \textbf{GD98\_b} & 121 & 2.182 & 75 & 97.702 & 96.281 & 98.014 & 96.785 \\ \hline
\textbf{30} & \textbf{Journals} & 124 & 96.323 & 0 & 94.443 & 94.443 & 95.951 & 95.951 \\ \hline
\textbf{31} & \textbf{GD96\_d} & 180 & 2.533 & 58 & 92.007 & 91.094 & 92.831 & 92.011 \\ \hline
\textbf{32} & \textbf{GD01\_a} & 311 & 4.116 & 121 & 97.743 & 97.343 & 98.169 & 97.845 \\ \hline
\textbf{33} & \textbf{USAir97} & 332 & 12.807 & 135 & 98.901 & 98.685 & 99.189 & 99.029 \\ \hline
\textbf{34} & \textbf{GD00\_a} & 352 & 2.182 & 177 & 98.588 & 98.112 & 98.835 & 98.442 \\ \hline
\textbf{35} & \textbf{SmallW} & 396 & 5.020 & 228 & 98.522 & 97.791 & 99.244 & 98.871 \\ \hline
\textbf{36} & \textbf{GD97\_c} & 452 & 2.035 & 395 & 99.893 & 99.548 & 99.920 & 99.661 \\ \hline
\textbf{37} & \textbf{Erdos971} & 472 & 5.568 & 168 & 97.643 & 97.303 & 98.130 & 97.860 \\ \hline
\textbf{38} & \textbf{Erdos981} & 485 & 5.695 & 171 & 97.690 & 97.363 & 98.151 & 97.890 \\ \hline
\textbf{39} & \textbf{Erdos991} & 492 & 5.760 & 173 & 97.622 & 97.287 & 98.213 & 97.962 \\ \hline
\textbf{40} & \textbf{GD00\_c} & 638 & 3.197 & 273 & 97.729 & 97.221 & 98.391 & 98.031 \\ \hline
\textbf{41} & \textbf{GD01\_Acap} & 953 & 1.343 & 763 & 99.595 & 98.871 & 99.723 & 99.228 \\ \hline
\textbf{42} & \textbf{Roget} & 1022 & 7.139 & 86 & 95.911 & 95.882 & 96.894 & 96.872 \\ \hline
\textbf{43} & \textbf{SmaGri} & 1059 & 9.284 & 251 & 97.350 & 97.193 & 98.118 & 98.006 \\ \hline
\textbf{44} & \textbf{GD96\_a} & 1096 & 3.060 & 4 & 91.707 & 91.707 & 93.822 & 93.822 \\ \hline
\textbf{45} & \textbf{GD06\_Java} & 1538 & 10.165 & 394 & 95.931 & 95.645 & 97.439 & 97.259 \\ \hline
\textbf{46} & \textbf{Csphd} & 1882 & 1.849 & 1306 & 99.481 & 99.000 & 99.548 & 99.129 \\ \hline
\textbf{47} & \textbf{Yeast} & 2361 & 5.630 & 952 & 98.288 & 97.956 & 98.742 & 98.498 \\ \hline
\textbf{48} & \textbf{ODLIS} & 2909 & 11.260 & 567 & 96.305 & 96.160 & 97.321 & 97.215 \\ \hline
\textbf{49} & \textbf{SciMet} & 3084 & 6.744 & 894 & 97.687 & 97.475 & 98.306 & 98.151 \\ \hline
\textbf{50} & \textbf{Kohonen} & 4470 & 5.690 & 1671 & 96.158 & 95.534 & 97.365 & 96.937 \\ \hline
\textbf{51} & \textbf{EPA} & 4772 & 3.734 & 2716 & 98.679 & 98.047 & 99.125 & 98.706 \\ \hline
\textbf{52} & \textbf{UspowerGrid} & 4941 & 2.669 & 1447 & 95.819 & 95.427 & 96.576 & 96.255 \\ \hline
\textbf{53} & \textbf{Erdos972} & 5488 & 2.582 & 4321 & 99.713 & 99.246 & 99.796 & 99.464 \\ \hline
\textbf{54} & \textbf{Erdos982} & 5822 & 2.533 & 4647 & 99.718 & 99.222 & 99.812 & 99.483 \\ \hline
\textbf{55} & \textbf{Erdos992} & 6100 & 2.464 & 4911 & 99.713 & 99.184 & 99.822 & 99.495 \\ \hline
\textbf{56} & \textbf{Zewail} & 6752 & 16.049 & 827 & 96.440 & 96.385 & 97.203 & 97.161 \\ \hline
\textbf{57} & \textbf{Erdos02} & 6927 & 2.446 & 5640 & 95.243 & 85.887 & 98.308 & 94.981 \\ \hline
\textbf{58} & \textbf{Geom} & 7343 & 3.241 & 5677 & 99.379 & 98.456 & 99.678 & 99.200 \\ \hline
\textbf{59} & \textbf{EVA} & 8497 & 1.580 & 7656 & 99.847 & 99.188 & 99.873 & 99.323 \\ \hline
\textbf{60} & \textbf{Lederberg} & 8843 & 9.393 & 2417 & 96.111 & 95.797 & 97.148 & 96.918 \\ \hline
\textbf{61} & \textbf{California} & 9664 & 3.305 & 5958 & 98.866 & 98.170 & 99.270 & 98.823 \\ \hline
\textbf{62} & \textbf{FA} & 10617 & 12.016 & 4235 & 98.324 & 98.007 & 98.715 & 98.472 \\ \hline
\textbf{63} & \textbf{foldoc} & 13356 & 13.697 & 1 & 94.982 & 94.982 & 96.410 & 96.410 \\ \hline
\textbf{64} & \textbf{EAT\_RS} & 23219 & 26.266 & 2860 & 97.376 & 97.336 & 98.072 & 98.043 \\ \hline
\textbf{65} & \textbf{EAT\_SR} & 23219 & 26.266 & 2861 & 97.364 & 97.324 & 98.071 & 98.041 \\ \hline
\textbf{66} & \textbf{Dictionary28} & 52652 & 3.382 & 27847 & 98.498 & 97.915 & 98.778 & 98.303 \\ \hline
\textbf{67} & \textbf{Wordnet3} & 82670 & 2.913 & 45223 & 97.687 & 96.700 & 98.070 & 97.246 \\ \hline \hline
&\textbf{Average Efficiency}&&&& \textbf{97.302} & \textbf{96.586}& \textbf{98.070} & \textbf{97.632}\\ \hline
&\textbf{Standard Deviation}&&&&\textbf{1.833}& \textbf{2.384}& \textbf{1.380}& \textbf{1.447}\\
%\hline &&&&&&&&\\ \hline \hline
\hline \hline
1 & \textbf{GD06\_theory} & 101 & 3.762 & 0 & 20.875 & 20.875 & 20.222 & 20.222 \\ \hline
2 & \textbf{GD96\_b} & 111 & 3.477 & 0 & 84.842 & 84.842 & 85.789 & 85.789 \\ \hline
3 & \textbf{GD98\_c} & 112 & 3.000 & 0 & 62.825 & 62.825 & 66.599 & 66.599 \\ \hline
\end{tabular}}
\label{r-tab5}
\end{table}
\end{center}

%\pagebreak
\section{Conclusion and Further Work}

In this paper, we coin a new problem called betweenness-ordering-problem and address its importance with real-world examples and provide a feasible and practical heuristic to solve it. According to our problem statement, betweenness-ordering problem refers to the ordering of two nodes. We extend the heuristic to solve the generic version of the betweenness-ordering problem for $k$ nodes when $k\ll$ total number of nodes. When $k$ is $O(n)$, our heuristic will take a lot of extra time and is not recommended to be used. The heuristic is partially based on the analysis of random $G(n,p)$ graphs which are distinct from the real-world graphs found in nature. Therefore, we perform an extensive testing of the proposed heuristic on a broad range of $70$ real-world networks from the SNAP dataset \cite{snap} to give an experimental evidence of the heuristic's worthiness. To the best of our knowledge, this is the first of its kind study addressing the ``ordering problem" in centrality measures. While our work is a first attempt to provide a solution to the centrality-ordering-problem for the betweenness measure, this should lead to the asking and answering of this question across several popular measures that have seen its applications in diverse areas.

\begin{itemize}
\item Our model performs very well on both real and synthetic networks. The formulation of EDDBM is based on the analysis of random graphs. Random graphs do not possess high clustering coefficient, and thus this model does not perform well on the graphs with high clustering coefficient. In highly clustered graphs, a better model is desirable. An interesting problem would be to tune BOLT so that the clustering has no effect on the results. 

\item Theoretically bound on the EDDBM's error in approximating the optimal sampling probabilities is still open. Coming up with a better model than EDDBM will increase the efficiency of BOLT for betweenness-ordering which is a potential future direction.

\item One can attempt to ask a similar question for the Pagerank-ordering-problem, Closeness-ordering-problem or any other centrality-ordering-problem. Any attempt to address these problems in the same spirit as our addressing the betweenness-ordering-problem would collectively be a significant contribution to applied sciences where centrality measures are being increasingly applied.
\end{itemize}

%\ifCLASSOPTIONcompsoc
%  % The Computer Society usually uses the plural form
%\else
%  % regular IEEE prefers the singular form
\section*{Acknowledgment}
%\fi
The authors would like to thank the IIT Ropar HPC committee for providing the resources for performing experiments. They also would like to thank S.R.S. Iyengar and the Malgudi team at IIT Ropar for their comments to improve the presentation of the paper.

%\ifCLASSOPTIONcaptionsoff
%  \newpage
%\fi



%\bibliographystyle{IEEEtran}


\bibliographystyle{plainnat}
\bibliography{ref1}
\pagebreak
\appendix

\section{Eccentricity Ordering}
We illustrate centrality ordering problem in the context of a centrality measure called the \textit{eccentricity} measure and give a simple approximation approach for eccentricity ordering. 

\begin{figure}[h]
\centering
\subfigure[1]{%
\includegraphics[width=.35\textwidth]{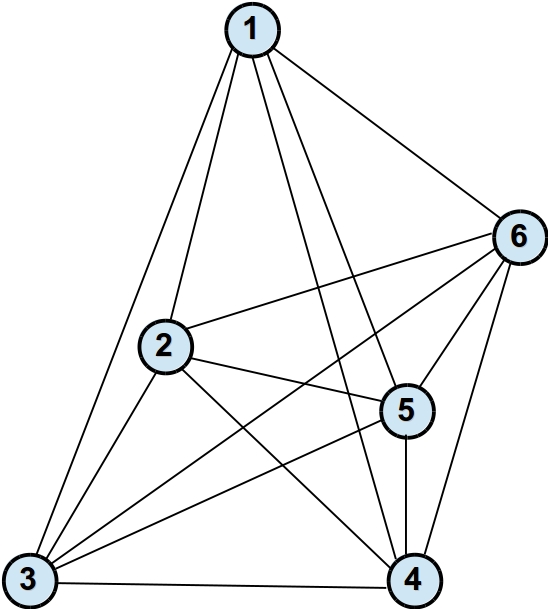}
\label{fig0.1}}
\quad
\subfigure[2]{%
\includegraphics[width=.35\textwidth]{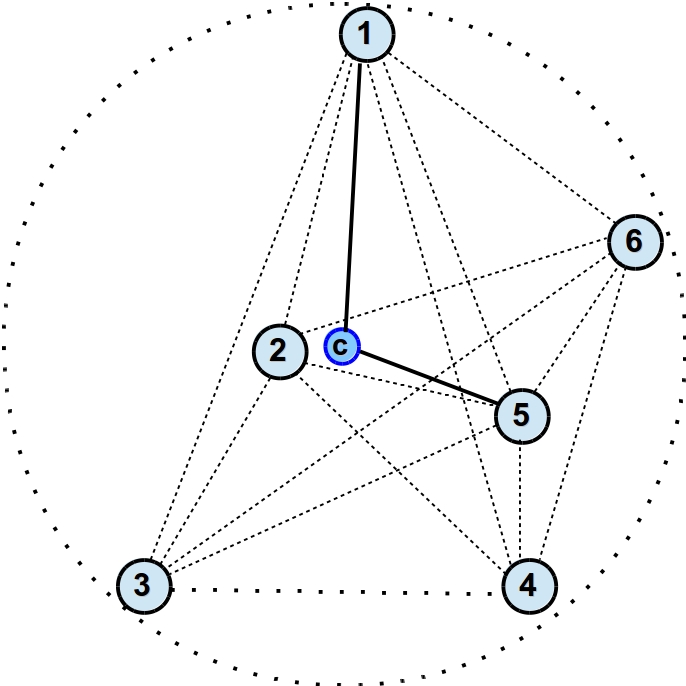}
\label{fig0.2}}
\caption{Eccentricity ordering in 2-D Euclidean plane}
\label{fig0}
\end{figure}  

Eccentricity of a node $v$\cite{Hage95} in a connected graph $G$ is defined as the shortest distance to the farthest node from $v$ in $G$. Center of a graph which is a solution to the \textit{facility location problems}, is calculated by picking the nodes with least eccentricity. Finding eccentricity of all nodes is as expensive as finding closeness, betweenness or stress centrality for all nodes in respect of time. Given a graph in the two dimensional euclidean space, if we were to solve the \emph{eccentricity ordering problem} of two nodes in that graph without computing the eccentricities, we would go about the following way: drawing a minimum circle (Disk) covering all the nodes is a very well known problem called \textit{smallest-circle problem} or \textit{minimum covering circle problem}. A linear ($O(n)$) time randomized algorithm by Welzl~\cite{Welzl} can find the smallest circle covering $n$ points on a 2-D euclidean plane. Once we find the smallest circle, an approximate solution to the eccentricity comparison problem is to compare the distance from the center of smallest circle to the nodes. If the nodes are evenly distributed in the smallest circle, then the node closer to the center of that smallest circle is likely to have smaller eccentricity and vice versa. Therefore, the eccentricity ordering problem can be estimated in linear time as opposed to finding it the conventional way by considering all possible distances from the given vertex to all other vertices. A theoretical bound on the ordering efficiency of the above scheme is still open. 
\section{•}
\subsection{Corroborating Experimental Evidence of Linear Time Running for Efficient Ordering}
In this section, we show experimentally that the expected time for betweenness-ordering heuristic is linear with the number of edges ($m$). If the average degree is constant, then ordering time also shown linear to the number of nodes ($n$). We have picked Gnutella-family of networks from \cite{snap}. This family comprises of 9 different snapshot of Gnutella peer-to-peer file sharing network from August 2002. The details of the networks and betweenness-ordering results on the networks are summarized in table \ref{tabto}. The columns in the table contain name of the network, number of nodes ($n$), average degree of nodes in the network(AVG.D.), average ordering time over 500 random pairs of nodes (Avg. time) in seconds, and average efficiency in ordering the 500 randomly picked pairs of nodes based on their betweenness score (Avg. Efficiency) respectively. The average betweenness-ordering  for all the networks in this network-family is above 97\% which is surely a high accuracy in very less ordering time.

\begin{table*}[htbp]
\centering
\caption{Average betweenness-ordering time and efficiency in Gnutella-family of networks from \cite{snap}}
\resizebox{.6\textwidth}{!}
{
\begin{tabular}{|l|r|r|r|r|c|}
\hline
Instance & \multicolumn{1}{l|}{n} & \multicolumn{1}{l|}{Avg. D.} & \multicolumn{1}{l|}{m} & \multicolumn{1}{l|}{Avg. time} & \multicolumn{1}{l|}{Avg. Efficiency} \\ \hline \hline
p2p-Gnutella08 & 6301 & 6.59483 & 41554 & 0.305 & 98.2 \\ 
P2p-Gnutella09 & 8114 & 6.41188 & 52026 & 0.406 & 97.6 \\ 
P2p-Gnutella06 & 8717 & 7.23299 & 63050 & 0.490 & 97.8 \\ 
P2p-Gnutella05 & 8846 & 7.19851 & 63678 & 0.479 & 98 \\ 
P2p-Gnutella04 & 10876 & 7.35454 & 79988 & 0.604 & 98.2 \\ 
P2p-Gnutella25 & 22687 & 4.82259 & 109410 & 1.055 & 98.6 \\ 
P2p-Gnutella24 & 26518 & 4.93016 & 130738 & 1.198 & 98.4 \\ 
P2p-Gnutella30 & 36682 & 4.81588 & 176656 & 1.671 & 98.2 \\ 
P2p-Gnutella31 & 62586 & 4.72604 & 295784 & 2.795 & 97.2 \\ \hline
\end{tabular}}
\label{tabto}
\end{table*}

The plots on the result data are shown in Figure~\ref{tordering}. x-axis represents either the number of nodes or the number of edges and y-axis denotes the average time taken for ordering two nodes based on their betweenness score. It is clear from the plots that the average betweenness-ordering time is linear with the number of edges. Here ordering time is also linear with the number of edges. It is due to the constant and small average degree of nodes in all of the networks in considered network-family which made number of edges as linear to the number of nodes. These plots are the corroborating experimental evidence of the heuristic's (BOLT's) linear running time on real-world graphs. This result also concludes that scaling of graphs in the  terms of number of edges will linearly affect the ordering time and thus our ordering algorithm will be very quick and accurate than any other possible algorithms. 

\begin{figure}[H]
\centering
\subfigure[Number of nodes Vs Ordering time]{%
\includegraphics[width=.45\textwidth]{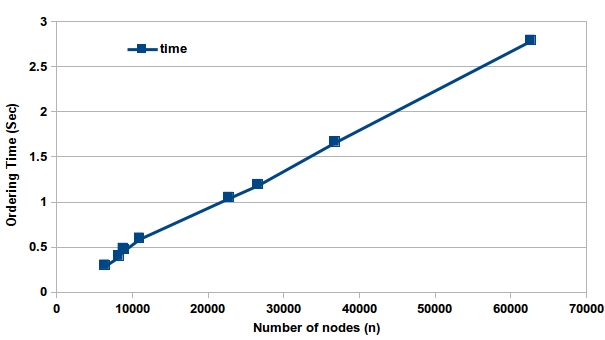}
\label{cor_syn}}
\subfigure[Number of edges Vs Ordering time]{%
\includegraphics[width=.45\textwidth]{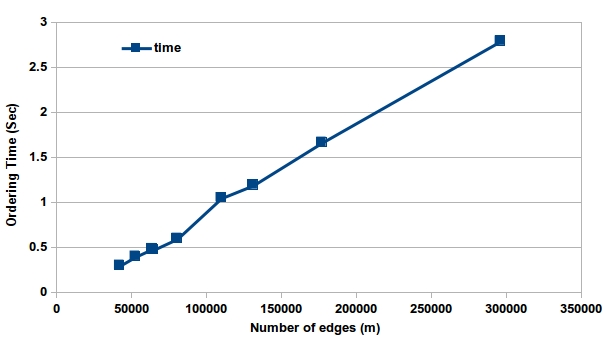}
\label{cor_real}}
\caption{Average betweenness-ordering time Vs size of the graph in terms of number of node  and number of edges on Gnutella-family of networks}
\label{tordering}
\end{figure}

\subsection{Average Error and Efficiency vs Number of Sampled Nodes ($T$)\label{samples}}
In each of the Figure~\ref{iter-error} and Figure~\ref{iter-order}, there are two plots. The first plot is for the change in the average error vs $T$ and the second plot is for the change in the efficiency vs $T$ on considered synthetic graphs of size 10000 that are mentioned in Table \ref{tab_synthetic}.

\begin{figure}[htb]
\centering
\subfigure[Average Error]{%
\includegraphics[width=.48\columnwidth]{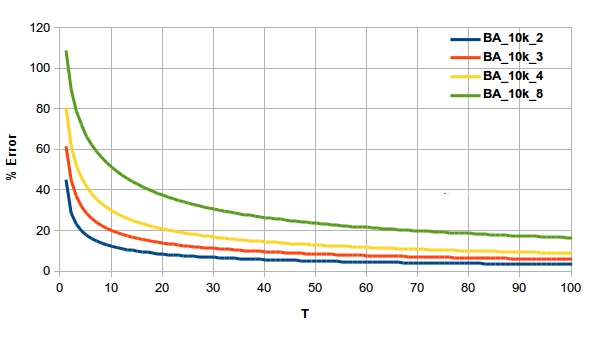}
\label{iter-error-ba}}
\subfigure[Average Efficiency]{%
\includegraphics[width=.48\columnwidth]{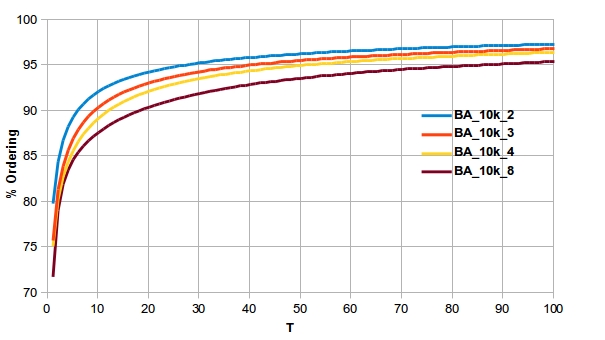}
\label{iter-error-er}}
\caption{Iterative performance of our approach on considered BA\_10k type synthetic graphs. }
\label{iter-error}
\end{figure} 

\begin{figure}[htb]
\centering
\subfigure[Average Error]{%
\includegraphics[width=.48\columnwidth]{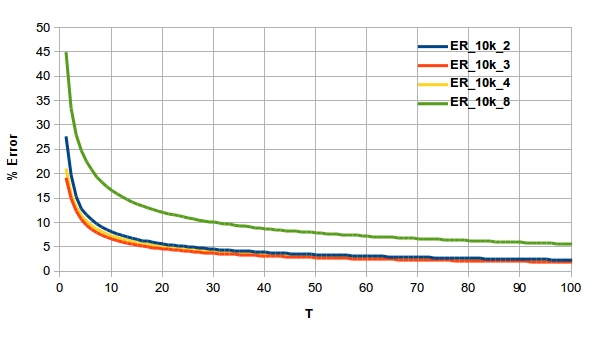}
\label{iter-order-er}}
\subfigure[Average Efficiency]{%
\includegraphics[width=.48\columnwidth]{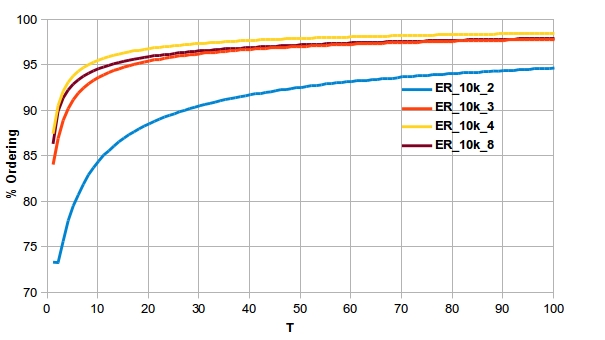}
\label{iter-order-er}}
\caption{Iterative performance of our approach on considered ER\_10k type synthetic graphs.}
\label{iter-order}
\end{figure}

\subsection{Average Relaxed efficiency ($\xi^t$) vs $t$ \label{relaxed_ef}}
In this section we will inspect the plot between the relaxed efficiency $\xi^t$ and $t$ on some synthetic networks of \textbf{1k} nodes. Similar results are achieved on all considered networks, but due to page limit, we skip the plots for other networks. The relaxed efficiency is computed by the formula given in section \ref{eff}. We vary $t$ for $t$=2,3,5,10. The plots are compiled in Figure~\ref{re}. For each $t$, we generate 5 synthetic networks and then calculate the relaxed efficiency on each graph and average them to get average relaxed efficiency. We plotted average relaxed efficiency for both type of considered synthetic graphs (ER and BA). 
BO, 2-E, 3-E, 5-E, 10-E are the labels assumed for average relaxed efficiency at $t$=0,2,3,5,10 respectively. The plots demonstrate that the efficiency increases with an increase in $t$, i.e., in real-world situation where relaxation is allowed in ordering, BOLT will perform much better than its usual performance. The relaxation in ordering means that betweenness-ordering of the nodes with approximately same betweenness rank is ignored and only the betweenness-ordering of nodes with difference greater than the threshold value (t) in their betweenness ranks are considered. Thus, we can say that BOLT results ordering very close to the exact betweenness-ordering.
\begin{figure}[htb]
\centering
\includegraphics[width=.7\columnwidth]{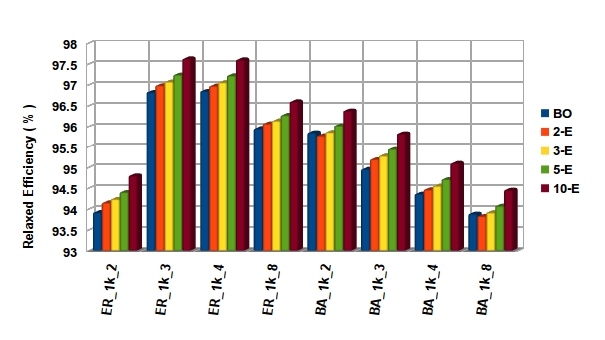}
\caption{Average relaxed efficiency for different values of $t$ in some synthetic networks }
\label{re}
\end{figure}

%\begin{IEEEbiographynophoto}{Shubham Chaudhary}

%\newpage

%\begin{IEEEbiographynophoto}{S.R.S. Iyengar}


%\vfill

\end{document}